\documentclass[10pt,journal,onecolumn,a4paper]{IEEEtran}
\usepackage{amsfonts,amssymb,amsthm,amstext}
\usepackage{cite,color,dsfont,enumerate,hyperref,makeidx,multicol,verbatim,xfrac,xr,url}
\usepackage[cmex10]{amsmath}
\usepackage{pict2e}
\usepackage[mathscr]{euscript}
\usepackage[geometry]{ifsym}
\hypersetup{
pdftitle		={A Simple Derivation of the Refined Sphere Packing Bound Under Certain Symmetry Hypotheses},
pdfauthor		={Bar{\i}s Nakiboglu},
pdfsubject		={Refined SPB},
pdfkeywords		={-}, 
}
\hypersetup{hidelinks}
\title{A Simple Derivation of the Refined Sphere Packing Bound Under Certain Symmetry Hypotheses}
\author{Bar\i\c{s} Nakibo\u{g}lu\\ \small{\href{mailto:bnakib@metu.edu.tr}{bnakib@metu.edu.tr}}
\thanks{This work is supported 
by  The Science Academy, Turkey, under The
Science Academy's Young Scientist Awards Program (BAGEP)
and by The Scientific and Technological Research Council 
of Turkey (T\"{U}B\.{I}TAK) under Grant 119E053.
}\thanks{
This paper was presented in part at 
the 2019 IEEE International Symposium on Information Theory   
\cite{nakiboglu19-ISIT}.}
\vspace{-0.28cm}
}
\newcommand{\set} [1]			{{\mathscr{{#1}}}}
\newcommand{\alg}[1]			{{\mathcal{{#1}}}}
\newcommand{\rndv}[1]      {{\mathsf{{#1}}}}

\newcommand{\msr}[1]       {{\it    {{#1}}}}
\newcommand{\cnst}[1]      {{\mathit{{#1}}}}

\newcommand{\integers}[1]	{{\mathbb{Z}}_{^{{#1}}}}

\newcommand{\reals}[1]		{{\mathbb{R}}_{^{{#1}}}}
\newcommand{\bigo} [1]     {{\cnst{O}\left({{#1}}\right)}}
\newcommand{\smallo}[1]    {{\cnst{o}\left({{#1}}\right)}}
\newcommand{\bigtheta}[1]  {{\cnst{\Theta}\left({{#1}}\right)}}

\renewcommand{\vec}[1]     {\overrightarrow{{#1}}}

\newcommand{\dif}[1]       {{\mathrm{d}{#1}}}  
\newcommand{\der}[2]        {\tfrac{\dif{#1}}{\dif{#2}}}  
\newcommand{\pder}[2]       {\tfrac{\partial{#1}}{\partial{#2}}}

\newcommand{\DEF}[0]			{{\!\!~\triangleq\!~}}  
\newcommand{\mtimes}[0]			{{\circledast}}

\newcommand{\AC}[0]            {{\prec}}  
  
\newcommand{\NAC}[0]           {{\nprec}}  
\newcommand{\abs}[1]           {{\left\lvert{{#1}}\right\lvert}}

\newcommand{\lon}[1]           {{{\left\lVert{{#1}}\right\lVert}}} 

\newcommand{\IND}[1]           {{\mathds{1}_{\{#1\}}}}    
\newcommand{\ind}[0]           {{\imath}}
\newcommand{\jnd}[0]           {{\jmath}}
\newcommand{\knd}[0]           {{\kappa}}
\newcommand{\tin}[0]           {{\cnst{t}}}
\newcommand{\blx}[0]           {{\cnst{n}}}

\newcommand{\PXS}[2]         {{\bf P}_{{#1}}\!\left[{#2}\right]}
\newcommand{\EXS}[2]         {{\bf E}_{{#1}}\!\left[{#2}\right]}
\newcommand{\VXS}[2]         {{\bf V}_{{#1}}\!\left[{#2}\right]}

\newcommand{\PX}[1]          {\PXS{\!}{{#1}}}                      
\newcommand{\EX}[1]          {\EXS{\!}{{#1}}}                      

\newcommand{\fX}[0]          {{\cnst{f}}}

\newcommand{\gX}[0]          {{\cnst{g}}}   
\newcommand{\GX}[0]          {{\cnst{G}}}   
\newcommand{\hX}[0]          {{\cnst{h}}}

\newcommand{\SGEX}[0]        {{\cnst{E}_{\cnst{L}}}}

\newcommand{\RD}[3]				
{{\cnst{D}}_{{#1}}            \!\left(\left.            \! {#2}\right\Vert {#3}                  \right)}
\newcommand{\CRD}[4]			
{{\cnst{D}}_{{#1}}            \!\left(\left.\!\left.    \! {#2}\right\Vert {#3} \right\vert{{#4}}\right)}

\newcommand{\RMI}[3]	{{\cnst{I}}_{{#1}}        \!\left(        \! {#2};  \!{#3}                \!\right)}

\newcommand{\CX}[1]				{{\cnst{C}}_{{#1}}}
\newcommand{\RC}[2]				{{\cnst{C}}_{{#1},{#2}}}

\newcommand{\CRC}[3]			{{\cnst{C}}_{{#1},{#2},{#3}}}

\newcommand{\htdelta}[0] {{\cnst{\varDelta}}}
\newcommand{\composition}[0]	{{\cnst{\varUpsilon}}}
\newcommand{\costf}[0]			{{\cnst{\rho}}}
\newcommand{\costc}[0]			{{\cnst{\varrho}}}

\newcommand{\rfm}[0]			{{{\msr{\nu}}}}

\newcommand{\cset}[0]			{{\set{A}}}

\newcommand{\dspe}[1]			{{\cnst{E}_{sp\!}'}       \left({#1}\right)}

\newcommand{\spe}[1]			{{\cnst{E}_{sp\!}}       \left({#1}\right)}

\newcommand{\rate}[0]			{{\cnst{R}}}
\newcommand{\GCD}[1]			{{{{\cnst{\Phi}}}\left({{#1}}\right)}}

\newcommand{\cln}[1]          {{{\xi}_{{#1}}}}
\newcommand{\cla}[2]          {{{\xi}_{{#1}}^{{#2}}}}

\newcommand{\rno}[0]          {{\cnst{\alpha}}}
\newcommand{\rnt}[0]          {{\cnst{\eta}}}
\newcommand{\rns}[0]          {\rno^{\!\ast}}
                  
\newcommand{\Pe}[0]            {{\it P_{{{\bf e}}}}}      
                  
\newcommand{\Pem}[1]           {{\it P_{{{\bf e}}}^{{#1}}}}         
 
\newcommand{\enc}[0]           {{\varPsi}} 
\newcommand{\dec}[0]           {{\varTheta}}    

\newcommand{\brl}[0]           {{\alg{B}}}
\newcommand{\rborel}[1]        {{\brl}({#1})}

\newcommand{\oev}[0]           {{\set{E}}}

\newcommand{\GausDen}[1]		{{{{\cnst{\varphi}}}_{{#1}}}}

\newcommand{\pmea}[1]          {{{\alg{P}}({#1})}}

\newcommand{\pdis}[1]          {{{\set{P}}({#1})}}

\newcommand{\dinp}[0]          {{\cnst{x}}}

\newcommand{\inpS}[0]          {{\set{X}}}

\newcommand{\dout}[0]          {{\cnst{y}}}
\newcommand{\out}[0]           {{\rndv{Y}}}
\newcommand{\outS}[0]          {{\set{Y}}}
\newcommand{\outA}[0]          {{\alg{Y}}}

\newcommand{\dsta}[0]          {{\cnst{z}}}
\newcommand{\sta}[0]           {{\rndv{Z}}}

\newcommand{\dmes}[0]          {{\cnst{m}}}

\newcommand{\mesS}[0]          {{\set{M}}}

\newcommand{\estS}[0]          {{\widehat{{\set{M}}}}}

\newcommand{\mean}[0]        {{{\msr{\mu}}}}    
\newcommand{\mmn}[1]         {{{\mean}_{{#1}}}}

\newcommand{\mA}[0]				{{\msr{a}}}    
\newcommand{\amn}[1]			{{{\mA}_{{#1}}}}

\newcommand{\mB}[0]				{{\msr{b}}}    
\newcommand{\bmn}[1]			{{{\mB}_{{#1}}}}

\newcommand{\mP}[0]				{{\msr{p}}}

\newcommand{\mQ}[0]				{{\msr{q}}}    
\newcommand{\qmn}[1]			{{{\mQ}_{{#1}}}}

\newcommand{\Qm}[0]				{{{\cnst{Q}}}}

\newcommand{\mS}[0]				{{\msr{s}}}

\newcommand{\mU}[0]				{{\msr{u}}}

\newcommand{\Umn}[1]			{{{\cnst{U}}_{{#1}}}}

\newcommand{\mV}[0]				{{\msr{v}}}    
\newcommand{\vmn}[1]			{{{\mV}_{{#1}}}}

\newcommand{\mW}[0]				{{\msr{w}}}    
\newcommand{\wmn}[1]			{{{\mW}_{{#1}}}}
\newcommand{\wma}[2]			{{{\mW}_{{#1}}^{{#2}}}}
\newcommand{\Wm}[0]				{{{\cnst{W}}}}
\newcommand{\Wmn}[1]			{{{\cnst{W}}_{{#1}}}}
\newcommand{\Wma}[2]			{{{\cnst{W}}_{{#1}}^{{#2}}}}

\newcommand{\altug}[0]								{Altu\u{g}~}
\newcommand{\csiszar}[0]							{Csisz\'{a}r~}

\newcommand{\renyi}[0]								{R\'{e}nyi~}

\makeatletter
\DeclareRobustCommand{\bigplus}{%
	\mathop{\vphantom{\sum}\mathpalette\@bigplus\relax}\slimits@
}
\newcommand{\@bigplus}[2]{\vcenter{\hbox{\make@bigplus{#1}}}}
\newcommand{\make@bigplus}[1]{%
	\sbox\z@{$\m@th#1\sum$}%
	\setlength{\unitlength}{\wd\z@}%
	\begin{picture}(1.4,1.4)
	\linethickness{.17ex}
	\Line(.7,.14)(.7,1.26)
	\Line(.14,.7)(1.26,.7)
	\end{picture}%
}
\DeclareRobustCommand{\bigtimes}{%
	\mathop{\vphantom{\sum}\mathpalette\@bigtimes\relax}\slimits@
}
\newcommand{\@bigtimes}[2]{\vcenter{\hbox{\make@bigtimes{#1}}}}
\newcommand{\make@bigtimes}[1]{%
	\sbox\z@{$\m@th#1\sum$}%
	\setlength{\unitlength}{\wd\z@}%
	\begin{picture}(1,1)
	\linethickness{.17ex}
	\Line(.1,.1)(.9,.9)
	\Line(.1,.9)(.9,.1)
	\end{picture}%
}
\makeatother
\theoremstyle{plain}
\newtheorem{theorem}{Theorem}[section] 
\newtheorem{lemma}[theorem]{Lemma} 
\theoremstyle{definition}
\newtheorem{definition}[theorem]{Definition} 
\newtheorem{remark}[theorem]{Remark}

\numberwithin{equation}{section}

\begin{document}
\pagestyle{plain}
\pagenumbering{arabic}
\maketitle 
\thispagestyle{empty}
\begin{abstract}
	A judicious application of the Berry-Esseen theorem via suitable 
	Augustin information measures is demonstrated to be sufficient for deriving 
	the sphere packing bound with a prefactor that is 
	$\mathit{\Omega}\left(n^{-0.5(1-E_{sp}'(R))}\right)$ 
	for all codes on certain families of channels 
	---including the Gaussian channels and the non-stationary \renyi symmetric channels---
	and for the constant composition codes on stationary memoryless channels.  
	The resulting non-asymptotic bounds have definite approximation error terms.
	As a preliminary result that might be of interest on its own, the trade-off 
	between type I and type II error probabilities in the hypothesis testing 
	problem with (possibly non-stationary) independent samples is determined 
	up to some multiplicative constants, assuming that the probabilities of 
	both types of error are decaying exponentially with the number of samples, 
	using the Berry-Esseen theorem.	
\end{abstract}


\section{Introduction}\label{sec:introduction}
The decay of the optimal error probability 
with the block length for rates below the channel capacity 
has been studied since the early days of the information theory. 
For certain channels and for certain values of the rate, 
sharp bounds were found early on.
Elias in \cite{elias55B} for the binary symmetric channel, 
Shannon\footnote{The equivalence of the bounds in \cite{shannon59} to \eqref{eq:nonsingular} is proved in
	Appendix \ref{sec:ShannonsExpressions}.} 
in \cite{shannon59} for the additive Gaussian noise channel, and 
Dobrushin in \cite{dobrushin62B} 
for the strongly symmetric channels
---see the original publication in Russian to avoid typos present in the translation---
proved that
\begin{align}
\label{eq:nonsingular}
\Pem{(\blx)}
&=\bigtheta{\blx^{-\frac{1-\dspe{\rate}}{2}} e^{-\blx \spe{\rate}}}
&
&~
&&\forall\rate\in[\rate_{crit},\CX{}),
\end{align}
where\footnote{In this section, we suppress the dependence of 
	the sphere packing exponent to the channel in our notation and denote it by 
	\(\spe{\rate}\), rather than \(\spe{\rate,\!\Wm}\).} 
\(\amn{\blx}\!=\!\bigtheta{\bmn{\blx}}\) iff 
\(0<
\liminf\limits_{\blx\to\infty}\abs{\frac{\amn{\blx}}{\bmn{\blx}}}
\leq 
\limsup\limits_{\blx\to\infty}\abs{\frac{\amn{\blx}}{\bmn{\blx}}}
<\infty\),
\(\spe{\cdot}\) is the sphere packing exponent of the channel,
\(\dspe{\cdot}\) is its derivative with respect to the rate, 
\(\rate_{crit}\) is the rate at which the slope
of the sphere packing exponent curve is minus one, i.e.,\(\dspe{\rate_{crit}}=-1\),
and \(\CX{}\) is the capacity of the 
channel.
On the other hand, Elias proved in \cite{elias55B} 
for the binary erasure channels that  
\begin{align}
\label{eq:singular}
\Pem{(\blx)}
&=\bigtheta{\blx^{-\frac{1}{2}} e^{-\blx \spe{\rate}}}
&
&~
&
&~
&&\forall\rate\in[\rate_{crit},\CX{}).
\end{align}

Neither \eqref{eq:nonsingular}, nor \eqref{eq:singular}, holds for rates 
below the critical rate.
If, however, we replace the equality sign with the greater than or equal 
to sign, then both \eqref{eq:nonsingular} and  \eqref{eq:singular} hold 
for all rates below the channel capacity.
These lower bounds are customarily called sphere packing bounds (SPBs) 
because of the techniques used in their original derivations. 

Derivations of the SPB in \cite{elias55B,shannon59,dobrushin62B} relied 
on the geometric structure of the output space of the channel 
and parameters that can be defined only for some models.
The resulting bounds were expressed in terms of these parameters, as well.
Thus it was not even clear that SPBs in\cite{elias55B,shannon59,dobrushin62B} 
can be interpreted as specific instances of a general bound.
The evidence for such an interpretation came not from a breakthrough about 
the lower bounds on the error probability but 
from a breakthrough about the upper bounds.  
Gallager's seminal work \cite{gallager65} unified and generalized
the upper bounds on the error probability 
---at least in terms of the exponent--- 
in all the previous studies. 
It is only with Gallager's formulation in \cite{gallager65} 
that one can express the bounds in \cite{elias55B,shannon59,dobrushin62B} 
as \eqref{eq:nonsingular} and \eqref{eq:singular}.

The first complete proof of the SPB 
for arbitrary discrete stationary product 
channels\footnote{These channels are customarily called discrete 
	memoryless channels, i.e., DMCs. We call them DSPCs in order to underline 
	the stationarity of these channels and the absence of any constraints on 
	their input sets. 
	In principle, such constraints might exist and stationarity might be absent 
	in a discrete channel that is memoryless.}
(DSPCs) was presented in \cite{shannonGB67A}.
According to \cite[Thm. 2]{shannonGB67A} 
\begin{align}
\notag
\Pem{(\blx)}
&\geq e^{-\blx\left[\spe{\rate-\bigo{\blx^{-\sfrac{1}{2}}}}+\bigo{\blx^{-\sfrac{1}{2}}}\right]}
&
&~
&&\forall\rate\in(0,\CX{}),
\end{align}
where \(\amn{\blx}=\bigo{\bmn{\blx}}\) iff 
there exists a \(K\in\reals{+}\) such that \(\abs{\amn{\blx}}\leq K \bmn{\blx}\)
for all \(\blx\) large enough. 
In the following two years, the SPB was proved 
first for stationary product channels with finite input sets in \cite{haroutunian68}
and then for (possibly) non-stationary product channels in \cite{augustin69}.
Since then, the SPB has been proven for various channel models in 
\cite{ebert66,richters67,gallager,omura75,csiszarkorner,wyner88-b,burnashevK99,arikan02,lapidoth93,augustin78,nakiboglu19B,nakiboglu19E,nakiboglu17,nakiboglu18D,dalai13A,dalaiW17},
including some quantum information theoretic ones.
It is, however, worth noting that a general proof that holds 
for both Gaussian channels 
---considered in \cite{shannon59,ebert66,richters67}---
and for arbitrary DSPCs 
---considered in \cite{shannonGB67A}---
was absent until recently, see \cite{nakiboglu17} and \cite{nakiboglu18D}.
These later works on the SPB,
i.e., \cite{shannonGB67A,haroutunian68,augustin69,
	ebert66,richters67,gallager,omura75,csiszarkorner,wyner88-b,
	burnashevK99,arikan02,lapidoth93,augustin78,nakiboglu19B,nakiboglu19E,nakiboglu17,nakiboglu18D,
	dalai13A,dalaiW17},
were primarily interested in establishing the right exponential decay rate;
thus, they were content with prefactors of the form 
\(e^{-\smallo{{\blx}}}\),
where
\(\amn{\blx}=\smallo{\bmn{\blx}}\) iff 
for all \(\epsilon\!>\!0\) the inequality \(\abs{\amn{\blx}}\leq \epsilon \bmn{\blx}\)
holds for all \(\blx\) large enough. 
Some authors did obtain prefactors of the form 
\(e^{-\bigo{\sqrt{\blx}}}\) or 
\(e^{-\bigo{\ln \blx}}\), but obtaining the best possible
---if not tight--- 
prefactor was not an actual concern.

The quest for deriving SPBs with tight prefactors 
was put on the map again by \altug and Wagner in \cite{altugW11}
and \cite{altugW14A}.
According to \cite[Thm. 1]{altugW11} for any DSPC
with a Gallager symmetric\footnote{The condition 
	for Gallager symmetry is described
	in \cite[p. 94]{gallager}.
	The binary symmetric channel,
	the binary erasure channel,
	and channels considered in \cite{dobrushin62B} 
	are  symmetric according to this definition.} 
probability transition matrix \(\Wm\) with positive entries  
and rate \(\rate\) in \((0,\CX{})\), there exists an \(A\in\reals{+}\)
such that for any \(\epsilon>0\)
\begin{align}
\notag
\Pem{(\blx)}
&\geq A \blx^{-\frac{1-(1+\epsilon)\dspe{\rate}}{2}}e^{-\blx\spe{\rate}} 
&
&~
&&\forall \blx\geq \blx_{0}
\end{align}
for some \(\blx_{0}\) determined by \(\Wm\), \(\rate\), and \(\epsilon\).
The corresponding  result was established for the constant composition codes 
on arbitrary DSPCs in \cite[Thm. 1]{altugW14A}.
These results are generalized to classical-quantum 
channels in \cite[Thms. 8 and 14]{chengHT19},
with a slight improvement, allowing \(\epsilon=0\) 
for the symmetric channels. 

The primary tool for the derivations in \cite{altugW11,altugW14A,chengHT19}
is the Berry-Esseen theorem, albeit through certain auxiliary results
inspired by a theorem of Bahadur and Rao \cite{bahadurR60},
i.e., 
\cite[(74)]{altugW11}, \cite[Proposition 5]{altugW14A},
and \cite[Thm. 17]{chengHT19}.
Our main aim in this paper is to demonstrate that
the analysis can be simplified and the results can 
be strengthened and generalized through 
a more judicious application of 
the Berry-Esseen theorem via 
suitable Augustin information measures.

\cite{elias55B} and \cite{shannon59}
not only established \eqref{eq:nonsingular} and \eqref{eq:singular}
but also obtained closed-form expressions for the
upper and lower bounds implicit in \eqref{eq:nonsingular} and \eqref{eq:singular}.
Dobrushin went one step further and calculated the
exact asymptotic behavior of the SPB and the random coding bound 
by analyzing asymptotic behavior for
 the lattice and non-lattice cases separately
for random variables used to derive the SPB
and the random coding bound, see \cite[(1.32), (1.33), (1.34)]{dobrushin62B}.
Recently, the saddle point approximation is used to 
derive the SPB with the same asymptotic 
prefactor \cite[Corollary 2]{vazquezFKL18},
under weaker symmetry 
hypothesis,\footnote{The binary input Gaussian channel 
	and the binary erasure channel satisfy the symmetry hypothesis of  
	\cite{vazquezFKL18}, but not that of \cite{dobrushin62B}.} 
albeit by assuming a common support for all output distributions
of the channel and a non-lattice structure for the random variables 
involved.\footnote{Neither of these assumptions was needed while deriving
	this result in \cite{dobrushin62B}.} 
The main drawback of the analysis in \cite{vazquezFKL18} 
is the technical conditions that need to be confirmed for applying
the saddle point approximation via \cite[Proposition 2.3.1]{jensen}.
\begin{remark}
The proof of \cite[Corollary 2]{vazquezFKL18} holds only 
for channels whose Augustin center does not change with the order,
i.e., for \(\Wm\)'s for which \(\exists\mQ\) such that
\(\qmn{\rno,\Wm}\!=\!\mQ\) for all \(\rno\!\in\!(0,1)\).
Note that for the channels that violate this 
additional hypothesis,  the \(\bigo{\blx^{-\sfrac{1}{2}}}\) approximation 
error terms in \cite[(30)]{vazquezFKL18} are \(\rho\) dependent
because of the implicit \(\Qm\) dependence of  the
\(\bigo{\blx^{-\sfrac{1}{2}}}\) approximation 
error terms in \cite[(12) and (13)]{vazquezFKL18}.
In order to recover a result similar to \cite[Corollary 2]{vazquezFKL18} 
for a channel whose Augustin center changes with the order,
one needs a saddle point approximation 
that holds for a parametric family of i.i.d. sequences of random variables,
such as \cite[Proposition 1]{lanchoODKV19A},
rather than \cite[Lemma 2]{vazquezFKL18}, which holds for 
a single i.i.d sequence of random variables. 
\end{remark}

Let us finish this section with an overview of the paper.
In Section \ref{sec:model-and-notation}, we describe our model and notation.
In Section \ref{sec:HypothesisTesting}, first, we recall the connection between 
the hypothesis testing problem and the tilting, 
then we derive our primary technical tool using the Berry-Esseen theorem.
In Section \ref{sec:preliminary}, we review Augustin's information measures 
and the sphere packing exponent.
In Section \ref{sec:RSPB}, we state and prove refined SPBs for various models
using  Lemma \ref{lem:HTBE} and 
the observations recalled in Section \ref{sec:preliminary}.  
We conclude our presentation with a brief discussion
of the results, recent developments, and future work in Section \ref{sec:conclusion}.

\section{Model and Notation}\label{sec:model-and-notation}
For any set \(\inpS\), we denote the set of all probability mass functions that 
are non-zero only on finitely many elements of \(\inpS\) by \(\pdis{\inpS}\).
For any measurable space \((\outS,\outA)\), we denote the set of all probability 
measures on it by \(\pmea{\outA}\).
We denote the expected value of a measurable function \(\fX\) 
under the probability measure \(\mean\) by \(\EXS{\mean}{\fX}\).
Similarly, we denote the variance of \(\fX\) under \(\mean\),
i.e.,\(\EXS{\mean}{(\fX-\EXS{\mean}{\fX})^{2}}\),
by \(\VXS{\mean}{\fX}\). 

For sets \(\inpS_{1},\ldots,\inpS_{\blx}\) we denote
their Cartesian product by \(\inpS_{1}^{\blx}\)
and for \(\sigma\)-algebras \(\outA_{1},\ldots,\outA_{\blx}\)
we denote their product by \(\outA_{1}^{\blx}\).
We use the symbol \(\otimes\) to denote the product of measures.

A \emph{channel} \(\Wm\) is a function from \emph{the input set} \(\inpS\) to  the set of all probability 
measures on \emph{the output space} \((\outS,\outA)\):
\begin{align}
\notag
\Wm:\inpS \to \pmea{\outA}.
\end{align}
A channel \(\Wm\) is called a \emph{discrete channel} if both \(\inpS\) and \(\outA\) are finite sets.
The product of \(\Wmn{\tin}:\inpS_{\tin}\to\pmea{\outA_{\tin}}\) for \(\tin\in\{1,\ldots,\blx\}\)
is a channel of the form \(\Wmn{[1,\blx]}:\inpS_{1}^{\blx}\to\pmea{\outA_{1}^{\blx}}\) 
satisfying 
\begin{align}
\notag
\Wmn{[1,\blx]}(\dinp_{1}^{\blx})
&=\bigotimes\nolimits_{\tin=1}^{\blx}\Wmn{\tin}(\dinp_{\tin})
&
&\forall \dinp_{1}^{\blx}\in\inpS_{1}^{\blx}.
\end{align}
Any channel obtained by curtailing the input set of a length \(\blx\)
product channel is called a length \(\blx\) \emph{memoryless channel}.
A product channel \(\Wmn{[1,\blx]}\) is \emph{stationary} iff \(\Wmn{\tin}=\Wm\)
for all \(\tin\)'s for some \(\Wm\).
On a stationary channel, we denote the composition 
(i.e.,the empirical distribution, the type) 
of each \(\dinp_{1}^{\blx}\) by \(\composition(\dinp_{1}^{\blx})\);
thus \(\composition(\dinp_{1}^{\blx})\in\pdis{\inpS}\).

The pair \((\enc,\dec)\) is an \((M,L)\) \emph{channel code} on \(\Wmn{[1,\blx]}\) iff
\begin{itemize}
	\item The \emph{encoding function} \(\enc\) is a function from the message set \(\mesS\DEF\{1,2,\ldots,M\}\) to the input
	set \(\inpS_{1}^{\blx}\).
	\item The \emph{decoding function} \(\dec\) is a \(\outA_{1}^{\blx}\)-measurable function from 
	the output set \(\outS_{1}^{\blx}\) to the set \(\estS\DEF\{\set{L}:\set{L}\subset\mesS \mbox{~and~}\abs{\set{L}}\leq L\}\).
\end{itemize}
Given an \((M,L)\) channel code \((\enc,\dec)\) on \(\Wmn{[1,\blx]}\),
\emph{the conditional error probability} \(\Pem{\dmes}\) for \(\dmes \in \mesS\) 
and \emph{the average error probability} \(\Pem{}\)
are defined as
\begin{align}
\notag
\Pem{\dmes}
&\DEF \EXS{\Wmn{[1,\blx]}(\enc(\dmes))}{\IND{\dmes\notin\dec(\out_{1}^{\blx})}},
\\
\notag
\Pem{} 
&\DEF\tfrac{1}{M} \sum\limits_{\dmes\in \mesS} \Pem{\dmes}.
\end{align}
An encoding function \(\enc\) 
---hence the corresponding code---
on a stationary product channel, satisfies 
an empirical distribution constraint \(\cset\subset\pdis{\inpS}\) 
iff the composition of all of its codewords are in \(\cset\), 
i.e., iff \(\composition(\enc(\dmes))\in\cset\) for all \(\dmes\in\mesS\).
A code is called 
a constant composition code iff all of its codewords  have 
the same composition, i.e., there exists a \(\mP\) in \(\pdis{\inpS}\)
satisfying \(\composition(\enc(\dmes))=\mP\) for all \(\dmes\in\mesS\).

\section{Hypothesis Testing Problem and Berry-Esseen Theorem}\label{sec:HypothesisTesting}
Our primary aim in this section is to characterize ---up to a multiplicative constant--- 
the asymptotic behavior of type I error probability with the number of samples 
for a hypothesis testing problem between product measures, when type II error 
probability is decaying exponentially.
We use the Berry-Esseen theorem via the concepts of \renyi divergence and 
tilted probability measure to do that. 
Let us,  first, recall the definitions of the \renyi divergence and the tilted 
probability measure.

\begin{definition}\label{def:divergence}
For any \(\rno\in\reals{+}\) and \(\mW,\mQ\in\pmea{\outA}\),
\emph{the order \(\rno\) \renyi divergence between \(\mW\) and \(\mQ\)} is
\begin{align}
\notag
	\RD{\rno}{\mW}{\mQ}
	&\DEF \begin{cases}
	\tfrac{1}{\rno-1}\ln \int (\der{\mW}{\rfm})^{\rno} (\der{\mQ}{\rfm})^{1-\rno} \rfm(\dif{\dout})
	&\rno\neq 1\\
	\int  \der{\mW}{\rfm}\left[ \ln\der{\mW}{\rfm} -\ln\der{\mQ}{\rfm}\right] \rfm(\dif{\dout})
	&\rno=1
	\end{cases},
	\end{align}
	where \(\rfm\) is any measure satisfying \(\mW\AC\rfm\) and \(\mQ\AC\rfm\).
\end{definition}
The order one \renyi divergence is the Kullback-Leibler divergence.
For other orders, the \renyi divergence can be characterized in terms of the 
Kullback-Leibler divergence too: 
\begin{align}
\label{eq:varitional}
(1-\rno)\RD{\rno}{\mW}{\mQ}
&=\inf\nolimits_{\mV\in\pmea{\outA}} \rno \RD{1}{\mV}{\mW}+(1-\rno)\RD{1}{\mV}{\mQ}
\end{align}
with the convention that \(\rno \RD{1}{\mV}{\mW}+(1-\rno)\RD{1}{\mV}{\mQ}=\infty\) if
it would be otherwise undefined, see \cite[Thm. 30]{ervenH14}.
The characterization given in \eqref{eq:varitional} is related to another key concept 
of our analysis: the tilted probability measure. 
\begin{definition}\label{def:tiltedprobabilitymeasure}
	For any \(\rno\in\reals{+}\) and \(\mW,\mQ\in\pmea{\outA}\) satisfying 
	\(\RD{\rno}{\mW}{\mQ}<\infty\), 
	\emph{the order \(\rno\) tilted probability measure} \(\wma{\rno}{\mQ}\) is 
	\begin{align}
\notag
	\der{\wma{\rno}{\mQ}}{\rfm}
	&\DEF e^{(1-\rno)\RD{\rno}{\mW}{\mQ}}(\der{\mW}{\rfm})^{\rno} (\der{\mQ}{\rfm})^{1-\rno}.
	\end{align}
\end{definition}
If either \(\rno\) is in \((0,1)\) or \(\RD{1}{\wma{\rno}{\mQ}}{\mW}\) is finite,
then the tilted probability measure is the unique probability measure achieving 
the infimum  in \eqref{eq:varitional}
by \cite[Thm. 30]{ervenH14}, i.e.,
\begin{align}
\label{eq:varitional-tilted}
(1-\rno)\RD{\rno}{\mW}{\mQ}
&=\rno \RD{1}{\wma{\rno}{\mQ}}{\mW}+(1-\rno)\RD{1}{\wma{\rno}{\mQ}}{\mQ}.
\end{align}
Furthermore, under the same hypothesis the identities 
\begin{align}
\label{eq:CLQ}
\ln\der{\wma{\rno}{\mQ}}{\mQ}-\RD{1}{\wma{\rno}{\mQ}}{\mQ}
&=\rno\left(\ln\der{\wmn{ac}}{\mQ}-\EXS{\wma{\rno}{\mQ}}{\ln\der{\wmn{ac}}{\mQ}}\right), 
\\
\label{eq:CLW}
\ln\der{\wma{\rno}{\mQ}}{\mW}-\RD{1}{\wma{\rno}{\mQ}}{\mW}
&=(\rno-1)\left(\ln\der{\wmn{ac}}{\mQ}-\EXS{\wma{\rno}{\mQ}}{\ln\der{\wmn{ac}}{\mQ}}\right),
\end{align}
hold \(\wma{\rno}{\mQ}\)-a.s.,
where \(\wmn{ac}\) is the component of \(\mW\) that is absolutely continuous in \(\mQ\).

Let us proceed with recalling the Berry-Esseen theorem.
\begin{lemma}[\cite{berry41,esseen42,shevtsova10}]\label{lem:berryesseen}
	Let \(\{\cln{\tin}\}_{\tin\in \integers{+}}\) be independent  zero-mean
	random variables satisfying
	\(\sum\nolimits_{\tin=1}^{\blx} \EX{\cln{\tin}^{2}}<\infty\).
	Then there exists an absolute constant \(\omega\leq 0.5600\) satisfying 
	\begin{align}
	\notag
	\abs{\PX{\sum\nolimits_{\tin=1}^{\blx}\cln{\tin}< \tau}-\GCD{\tfrac{\tau}{\sqrt{\amn{2} \blx}}}}
	&\leq \omega\tfrac{\amn{3}}{\amn{2}\sqrt{\amn{2}\blx}}
	&
	&\forall \tau \in \reals{},
	\end{align}
	where \(\amn{\knd}=\frac{1}{\blx}\sum\nolimits_{\tin=1}^{\blx} \EX{\abs{\cln{\tin}}^{\knd}}\)
	and \(\GCD{\mS} =\tfrac{1}{\sqrt{2\pi}} \int_{-\infty}^{\mS} e^{-\sfrac{\dsta^2}{2}} \dif{\dsta}.\)
\end{lemma}

Lemma \ref{lem:HTBE}, in the following, characterizes the trade-off between 
type I and type II error probabilities for a hypothesis testing problem with 
independent samples, assuming that both error probabilities are decaying 
---at least--- exponentially with the number of samples.
Lemma \ref{lem:HTBE}, which is derived using the Berry-Esseen theorem,
can be interpreted as a refinement of \cite[Thm. 5]{shannonGB67A},
which is derived using Chebyshev's inequality.

\begin{lemma}\label{lem:HTBE}
For any \(\rno\in(0,1)\), \(\blx\in\integers{+}\), 
\(\wmn{\tin},\qmn{\tin}\in\pmea{\outA_{\tin}}\), 
let \(\wmn{\tin,ac}\) be the component of \(\wmn{\tin}\) that is absolutely continuous in \(\qmn{\tin}\)
and let \(\amn{2}\), \(\amn{3}\), and \(\htdelta\) be
\begin{align}
\notag
\amn{2}
&\DEF\tfrac{1}{\blx}\sum\nolimits_{\tin=1}^{\blx} 
\EXS{\wma{\rno}{\mQ}}{\left(\ln\der{\wmn{\tin,ac}}{\qmn{\tin}}-\EXS{\wma{\rno}{\mQ}}{\ln\der{\wmn{\tin,ac}}{\qmn{\tin}}}\right)^{2}},
\\
\notag
\amn{3}
&\DEF\tfrac{1}{\blx}\sum\nolimits_{\tin=1}^{\blx} 
\EXS{\wma{\rno}{\mQ}}{\abs{\ln\der{\wmn{\tin,ac}}{\qmn{\tin}}-\EXS{\wma{\rno}{\mQ}}{\ln\der{\wmn{\tin,ac}}{\qmn{\tin}}}}^{3}},
\\
\notag
\htdelta
&\DEF \exp\left(
2\sqrt{2\pi e}\left(0.56\tfrac{\amn{3}}{\amn{2}}+\sqrt{\amn{2}}\right)
\right),
\end{align}
where \(\mW=\otimes_{\tin=1}^{\blx} \wmn{\tin}\) and \(\mQ=\otimes_{\tin=1}^{\blx} \qmn{\tin}\).
Then for any \(\oev\in\outA_{1}^{\blx}\) 
and \(\beta\geq\blx^{-\sfrac{1}{2}}e^{-\rno \sqrt{\amn{2}\blx}}\),
satisfying \(\mQ(\oev)\leq\beta e^{-\RD{1}{\wma{\rno}{\mQ}}{\mQ}}\),
we have
\begin{align}
\label{eq:lem:HTBE-converse}
\mW({\outS_{1}^{\blx}}\setminus\oev)
&\geq \htdelta^{\rno-1} \beta^{\frac{\rno-1}{\rno}}
\blx^{-\frac{1}{2\rno}}e^{-\RD{1}{\wma{\rno}{\mQ}}{\mW}}	
\end{align}
provided that \(\beta\leq\htdelta^{-\rno}\blx^{-\sfrac{1}{2}}e^{\rno \sqrt{\amn{2}\blx}}\).
Furthermore, for any \(\rno\in(0,1)\) and \(\beta\in\reals{+}\),
there exists an event \(\oev\in\outA_{1}^{\blx} \) such that
\begin{align}
\label{eq:lem:HTBE-achievability-q}
\mQ(\oev)
&\leq\beta e^{-\RD{1}{\wma{\rno}{\mQ}}{\mQ}},	
\\
\label{eq:lem:HTBE-achievability-w}
\hspace{-.3cm}
\mW({\outS_{1}^{\blx}}\setminus\oev)
&\leq 
\left[
\tfrac{1\vee\sqrt{8\pi \amn{2}}}{4\pi \amn{2}}
\ln\htdelta\right]^{\frac{1}{\rno}}
\tfrac{(\rno\beta)^{\frac{\rno-1}{\rno}}}{1-\rno}
\blx^{-\frac{1}{2\rno}} e^{-\RD{1}{\wma{\rno}{\mQ}}{\mW}}.
\end{align}
\end{lemma}

\begin{proof}[Proof of Lemma \ref{lem:HTBE}]
Let the random variables \(\cln{\tin}\) and \(\cln{}\) and the event \(\set{B}\) be
\begin{align}
\notag
\cln{\tin}
&\DEF\ln\der{\wmn{\tin,ac}}{\qmn{\tin}},
&
\cln{}
&\DEF \sum\nolimits_{\tin=1}^{\blx} \cln{\tin},
&
\set{B}
&\DEF \left\{\dout_{1}^{\blx}:\tau_0\leq \cln{}-\EXS{\wma{\rno}{\mQ}}{\cln{}}\leq \tau_1 \right\}. 
\end{align}
Thus \(\cln{}=\ln\der{\wmn{ac}}{\mQ}\) holds
\(\mQ\)-a.s. by the definitions of \(\cln{\tin}\) and \(\cln{}\).
Hence \eqref{eq:CLQ}, \eqref{eq:CLW},  and the definition of \(\set{B}\)
imply that
\begin{align}
	\notag
	\set{B}
	&=\left\{\dout_{1}^{\blx}:\rno\tau_0\leq
	\ln\der{\wma{\rno}{\mQ}}{\mQ}-
	\RD{1}{\wma{\rno}{\mQ}}{\mQ}
	\leq  \rno \tau_{1} 
	\right\},
	\\
	\notag
	&=\left\{\dout_{1}^{\blx}:
	(1-\rno)\tau_{0}\leq 
	\RD{1}{\wma{\rno}{\mQ}}{\mW}-\ln\der{\wma{\rno}{\mQ}}{\mW}
	\leq(1-\rno)\tau_{1} 
	\right\}.
	\end{align}
	Thus for any \(\oev\in \outA_{1}^{\blx}\), we have
	\begin{align}
	\label{eq:HTBE-1}
	\wma{\rno}{\mQ}(\oev \cap \set{B})
	&\leq \mQ(\oev) e^{\RD{1}{\wma{\rno}{\mQ}}{\mQ}+\rno\tau_{1}},
	\\
	\label{eq:HTBE-2}
	\mW({\outS_{1}^{\blx}}\setminus\oev)
	&\geq \wma{\rno}{\mQ}(\set{B} \setminus\oev )
	e^{-\RD{1}{\wma{\rno}{\mQ}}{\mW}+(1-\rno)\tau_{0}}.
	\end{align}

On the other hand, \(\cln{\tin}\)'s are jointly independent 
under the tilted probability measure \(\wma{\rno}{\mQ}\).
Thus the Berry-Esseen theorem, given in Lemma \ref{lem:berryesseen}, implies
	\begin{align}
	\notag
	\wma{\rno}{\mQ}(\set{B})
	&\geq \GCD{\tfrac{\tau_1}{\sqrt{\amn{2}\blx}}}-\GCD{\tfrac{\tau_0}{\sqrt{\amn{2}\blx}}}
	-2\tfrac{0.56}{\sqrt{\blx}} 
	\tfrac{\amn{3}}{\amn{2}\sqrt{\amn{2}}}
	\\
	\notag
	&=\tfrac{1}{\sqrt{2\pi}} \int_{\frac{\tau_{0}}{\sqrt{\amn{2}\blx}}}^{\frac{\tau_{1}}{\sqrt{\amn{2}\blx}}} e^{-\sfrac{\dsta^2}{2}} \dif{\dsta}
	-2\tfrac{0.56}{\sqrt{\blx}} \tfrac{\amn{3}}{\amn{2}\sqrt{\amn{2}}}
	\\
	\notag
	&\geq \tfrac{1}{\sqrt{2\pi}}e^{-\frac{(\abs{\tau_{0}}\vee \abs{\tau_{1}})^{2}}{2\blx\amn{2}}} 
	\tfrac{\tau_{1}-\tau_{0}}{\sqrt{\amn{2}\blx}}
	-2\tfrac{0.56}{\sqrt{\blx}} \tfrac{\amn{3}}{\amn{2}\sqrt{\amn{2}}}.
	\end{align}
	If we set 
	\(\tau_{0}=-\tfrac{2\ln\beta+\ln\blx}{2\rno}-\ln \htdelta\)
	and \(\tau_{1}=-\tfrac{2\ln\beta+\ln\blx}{2\rno}\), 
	then  \(-\sqrt{\amn{2}\blx}\leq \tau_{0}\leq \tau_{1}\leq \sqrt{\amn{2}\blx}\)
	by the hypothesis
	and
	\begin{align}
	\notag
	\wma{\rno}{\mQ}(\set{B})
	&\geq \tfrac{2}{\sqrt{\blx}}.
	\end{align}
	Furthermore, 
	\(\wma{\rno}{\mQ}(\oev \cap \set{B})\leq\tfrac{1}{\sqrt{\blx}}\)
	as a result of \eqref{eq:HTBE-1},
	\(e^{\rno\tau_{1}}=\tfrac{1}{\beta\sqrt{\blx}}\),  and 
	the hypothesis \(\mQ(\oev)\leq \beta e^{-\RD{1}{\wma{\rno}{\mQ}}{\mQ}}\).
	Thus \(\wma{\rno}{\mQ}(\set{B} \setminus\oev)\geq\tfrac{1}{\sqrt{\blx}}\).
	Then using  \eqref{eq:HTBE-2} and 
	\(e^{(1-\rno)\tau_{0}}=\beta^{\frac{\rno-1}{\rno}}\blx^{\frac{\rno-1}{2\rno}}\htdelta^{\rno-1}\)
	we get
	\begin{align}
	\notag
	\mW({\outS_{1}^{\blx}}\setminus\oev)
	&\geq \tfrac{1}{\sqrt{\blx}}
	e^{-\RD{1}{\wma{\rno}{\mQ}}{\mW}+(1-\rno)\tau_{0}}
	\\
	\notag
	&=\htdelta^{\rno-1}
	\beta^{\frac{\rno-1}{\rno}}
	\blx^{-\frac{1}{2\rno}}
	e^{-\RD{1}{\wma{\rno}{\mQ}}{\mW}}.
	\end{align}
\begin{remark}
	While deriving bounds similar to \eqref{eq:lem:HTBE-converse},
	the constants \(\tau_{0}\) and \(\tau_{1}\) are usually 
	assumed to satisfy \(\tau_{0}=-\tau_{1}\), see for example 
	\cite[Thm. 5]{shannonGB67A} or \cite[Thm. 11]{chengHT19}. 
	Such a choice, however, does not lead to tight bounds in our case. 
\end{remark}

To establish the existence of an event satisfying both 
\eqref{eq:lem:HTBE-achievability-q} and \eqref{eq:lem:HTBE-achievability-w}, 
let us consider the event \(\oev\) given in the following
\begin{align}
\label{eq:HTBE-C1}
\hspace{-.6cm}\oev
&\DEF\left\{\dout_{1}^{\blx}: \cln{}\geq\EXS{\wma{\rno}{\mQ}}{\cln{}}+\gamma \right\}
\cup \left\{\dout_{1}^{\blx}: \der{\mQ}{\rfm}=0 
~\&~ \der{\mW}{\rfm}>0\right\},
\end{align}
where \(\gamma\) is a real number to be determined later and
\(\rfm\) is any measure satisfying both
\(\mW\AC\rfm\) and \(\mQ\AC\rfm\).
\begin{remark}
The random variable \(\cln{}\) is defined only for \(\dout_{1}^{\blx}\)'s 
with a positive \(\der{\mQ}{\rfm}\). 
Thus one can define \(\cln{}\) to be infinite for \(\dout_{1}^{\blx}\)'s
satisfying both  \(\der{\mQ}{\rfm}=0\) and \(\der{\mW}{\rfm}>0\), 
and define the event \(\oev\) to be the event that \(\cln{}\) is greater than or equal 
to \(\EXS{\wma{\rno}{\mQ}}{\cln{}}+\gamma\). 	
\end{remark}
For the event \(\oev\) defined in \eqref{eq:HTBE-C1}, 
as a result of \eqref{eq:CLQ} we have
\begin{align}
\notag
\mQ(\oev)
&=\EXS{\mQ}{\IND{\cln{}\geq \EXS{\wma{\rno}{\mQ}}{\cln{}}+\gamma}}
\\
\notag
&=e^{-\RD{1}{\wma{\rno}{\mQ}}{\mQ}}
\EXS{\wma{\rno}{\mQ}}{\IND{\cln{}\geq \EXS{\wma{\rno}{\mQ}}{\cln{}}+\gamma}e^{-\rno (\cln{}-\EXS{\wma{\rno}{\mQ}}{\cln{}}) }}
\\
\label{eq:HTBE-C2}
&\leq e^{-\RD{1}{\wma{\rno}{\mQ}}{\mQ}-\rno \gamma}
\sum\nolimits_{\knd=0}^{\infty}
\wma{\rno}{\mQ}(\oev_{\knd}) e^{-\rno\knd},
\end{align}
where the event \(\oev_{\knd}\) is defined for each \(\knd\in\integers{}\)
to be
\begin{align}
\notag
\oev_{\knd}
&\DEF\{\dout_{1}^{\blx}:
\gamma+\knd\leq \cln{}-\EXS{\wma{\rno}{\mQ}}{\cln{}}<\gamma+\knd+1\}.
\end{align}
On the other hand, we can bound \(\wma{\rno}{\mQ}(\oev_{\knd})\)
uniformly for all integers \(\knd\)
using the Berry-Esseen theorem, i.e.,Lemma \ref{lem:berryesseen},
as follows
\begin{align}
\notag
\wma{\rno}{\mQ}(\oev_{\knd})
&\leq\GCD{\tfrac{\gamma+\knd+1}{\sqrt{\amn{2}\blx}}}
-\GCD{\tfrac{\gamma+\knd}{\sqrt{\amn{2}\blx}}}
+2\tfrac{0.56}{\sqrt{\blx}} 
\tfrac{\amn{3}}{\amn{2}\sqrt{\amn{2}}}
\\
\label{eq:HTBE-C3}
&\leq \tfrac{1}{\sqrt{\blx}}\left(\tfrac{1}{\sqrt{2\pi \amn{2}}}+2\tfrac{0.56\amn{3}}{\amn{2}\sqrt{\amn{2}}}\right). 
\end{align}
For \(\gamma=\tfrac{1}{\rno}\ln \left[\tfrac{1}{\sqrt{\blx}}\left(\tfrac{1}{\sqrt{2\pi \amn{2}}}+2\tfrac{0.56\amn{3}}{\amn{2}\sqrt{\amn{2}}}\right)
\tfrac{\beta^{-1}}{1-e^{-\rno}}\right]\),
\eqref{eq:lem:HTBE-achievability-q}
follows from \eqref{eq:HTBE-C2}, \eqref{eq:HTBE-C3}, 
and
\(\sum\nolimits_{\knd=0}^{\infty}e^{-\rno\knd}=\tfrac{1}{1-e^{-\rno}}\).

\(\mW({\outS_{1}^{\blx}}\setminus\oev)\) is bounded following a similar analysis,
by invoking \eqref{eq:CLW}, instead of \eqref{eq:CLQ}:
\begin{align}
\notag
\mW({\outS_{1}^{\blx}}\setminus\oev)
&=\EXS{\mW}{\IND{\cln{}<\EXS{\wma{\rno}{\mQ}}{\cln{}}+\gamma}}
\\
\notag
&=e^{-\RD{1}{\wma{\rno}{\mQ}}{\mW}}
\EXS{\wma{\rno}{\mQ}}{\IND{\cln{}<\EXS{\wma{\rno}{\mQ}}{\cln{}}+\gamma}e^{(1-\rno) (\cln{}-\EXS{\wma{\rno}{\mQ}}{\cln{}}) }}
\\
\notag
&\leq
e^{-\RD{1}{\wma{\rno}{\mQ}}{\mW}+(1-\rno)\gamma}
\sum\nolimits_{\knd=-1}^{-\infty}
\wma{\rno}{\mQ}(\oev_{\knd}) e^{(\knd+1)(1-\rno)}
\\
\notag
&\leq
e^{-\RD{1}{\wma{\rno}{\mQ}}{\mW}+(1-\rno)\gamma}
\tfrac{1}{1-e^{\rno-1}}
\sup\nolimits_{\knd\in\integers{-}} \wma{\rno}{\mQ}(\oev_{\knd}).
\end{align}
Invoking first \eqref{eq:HTBE-C3} 
and then \(\tfrac{1}{1-e^{\rno-1}}(1-e^{-\rno})^{\frac{\rno-1}{\rno}}\leq 
\tfrac{1}{1-\rno} \rno^{\frac{\rno-1}{\rno}} e^{\frac{1}{2\rno}} \)
we get,
\begin{align}
\notag
\mW({\outS_{1}^{\blx}}\setminus\oev)
&\leq
e^{-\RD{1}{\wma{\rno}{\mQ}}{\mW}}\beta^{\frac{\rno-1}{\rno}}
\tfrac{1}{1-e^{\rno-1}}(1-e^{-\rno})^{\frac{\rno-1}{\rno}}
\left[\tfrac{1}{\sqrt{\blx}}\left(\tfrac{1}{\sqrt{2\pi \amn{2}}}+2\tfrac{0.56\amn{3}}{\amn{2}\sqrt{\amn{2}}}\right)
\right]^{\frac{1}{\rno}}
\\
\notag
&\leq
e^{-\RD{1}{\wma{\rno}{\mQ}}{\mW}}
\tfrac{1}{1-\rno}(\beta\rno)^{\frac{\rno-1}{\rno}}
\left[\tfrac{2\sqrt{e}}{\sqrt{\blx \amn{2}}}\left(\tfrac{1}{\sqrt{8\pi}}+\tfrac{0.56\amn{3}}{\amn{2}}\right)\right]^{\frac{1}{\rno}}.
\end{align}
Then \eqref{eq:lem:HTBE-achievability-w} follows from 
\(\left(\tfrac{1}{\sqrt{8\pi}}+\tfrac{0.56\amn{3}}{\amn{2}}\right)
\leq
\tfrac{1\vee \sqrt{8 \pi  \amn{2}}}{8 \pi   \sqrt{\amn{2} e}}\ln \htdelta\).
\end{proof}

Lemma \ref{lem:HTBE} characterizes the asymptotic behavior 
of the trade-off between the optimal type I and type II error probabilities
for a hypothesis testing problem with independent samples:
\(\Pe_{II}^{(n)}\) is 
\(\bigtheta{\blx^{-\frac{1}{2\rno}}e^{-\RD{1}{\wma{\rno}{\mQ}}{\mW}}}\)
whenever
\(\Pe_{I}^{(n)}\) is 
\(\bigtheta{e^{-\RD{1}{\wma{\rno}{\mQ}}{\mQ}}}\).
For the stationary case, 
---i.e.,when \(\wmn{\tin}=\wmn{1}\), \(\qmn{\tin}=\qmn{1}\) 
for all \(\tin\)---
\csiszar and Longo \cite{csiszarL71} 
described how \eqref{eq:CLQ} and \eqref{eq:CLW}
can be used  together with an earlier result by Strassen,
\cite[Thm. 1.1]{strassen62},  
to characterize the exact 
asymptotic behavior of \(\Pe_{II}^{(n)}\) 
for the case when 
\(\Pe_{I}^{(n)}=e^{-\RD{1}{\wma{\rno}{\mQ}}{\mQ}}\),
i.e.,
\(\Pe_{II}^{(n)}=\left(K+\smallo{1}\right)
\blx^{-\frac{1}{2\rno}}e^{-\RD{1}{\wma{\rno}{\mQ}}{\mW}}\),
with some minor inaccuracies discussed in Remark \ref{rem:CL}.
One does not need to rely on \cite[Thm. 1.1]{strassen62} of Strassen
to characterize this exact asymptotic behavior. 
The Berry-Esseen theorem, however, is not sufficient 
for determining the value of the constant \(K\). 
In order to determine the constant \(K\), one needs to invoke 
either finer characterizations of the asymptotic behavior 
of sums of independent random
variables ---such as the ones in \cite[\S IV.2,\S IV.3]{esseen45}, 
\cite[\S42,\S43]{gnedenkoK54}---
or apply other techniques
---such as the saddle point approximation described in \cite[Prop. 2.3.1]{jensen}.
It is worth noting that both of these approaches
require hypotheses stronger than that of the Berry-Esseen theorem.
The situation is similar for other values of \(\rno\),
but of no interest for our discussion of the sphere
packing bound.

\begin{remark}\label{rem:CL}
	We believe the approach of \cite{csiszarL71} 
	is sound. Its calculations, however, seem to have some mistakes.
	Repeating the calculations as described in \cite{csiszarL71}, we recover
	the second line of \cite[(33)]{csiszarL71} as 
	\(\ln \tfrac{\rns}{1-\rns}-\frac{\ln S_{1}\sqrt{2\pi}}{\rns}
	+\smallo{1}\).
	With this modification \cite[Thm. 2]{csiszarL71} is consistent with 
	the intimately related results about the SPB proved earlier
	\cite[(1.32), (1.33)]{dobrushin62B} and since then
	\cite[(38)]{vazquezFKL18}.
\end{remark}

\section{Augustin's Information Measure and The Sphere Packing Exponent}\label{sec:preliminary}
The ultimate aim of this section is to define the sphere packing exponent
and review the properties of it that will be useful in 
our analysis.  
For that, we first recall the definitions of Augustin's information measures
and review their elementary properties.

\subsection{Augustin's Information Measures}\label{sec:informationmeasures}
Let us start by recalling the definition of the  conditional \renyi divergence.
\begin{definition}\label{def:conditionaldivergence}
For any \(\rno\in\reals{+}\), \(\Wm:\inpS\to\pmea{\outA}\), \(\mQ\in\pmea{\outA}\),
and \(\mP\in\pdis{\inpS}\) \emph{the order \(\rno\) conditional \renyi divergence for 
the input distribution \(\mP\)} is
\begin{align}
\notag
\CRD{\rno}{\Wm}{\mQ}{\mP}
&\DEF \sum\nolimits_{\dinp\in \inpS}  \mP(\dinp) \RD{\rno}{\Wm(\dinp)}{\mQ}.
\end{align}
\end{definition}

\begin{definition}\label{def:information}
	For any \(\rno\in\reals{+}\), \(\Wm:\inpS\to\pmea{\outA}\), and 
	\(\mP\in\pdis{\inpS}\)  
	\emph{the order \(\rno\) Augustin information for the input distribution \(\mP\)}
	is
	\begin{align}
\notag
	\RMI{\rno}{\mP}{\Wm}
	&\DEF \inf\nolimits_{\mQ\in \pmea{\outA}} \CRD{\rno}{\Wm}{\mQ}{\mP}.
	\end{align}
\end{definition}
The infimum is achieved by a unique probability 
measure\footnote{We refrain from including the channel symbol \(\Wm\) 
	in the symbol for the Augustin mean
	because the channel will be clear from the context.} \(\qmn{\rno,\mP}\),
called
\emph{the order \(\rno\) Augustin mean for the input distribution \(\mP\)},
by  \cite[{Lemma \ref*{C-lem:information}}]{nakiboglu19C}. 
Furthermore,
\begin{align}
\notag
\RD{1 \vee \rno}{\qmn{\rno,\mP}}{\mQ}
\geq
\CRD{\rno}{\Wm}{\mQ}{\mP}-\RMI{\rno}{\mP}{\Wm}
&\geq
\RD{1 \wedge \rno}{\qmn{\rno,\mP}}{\mQ}
\end{align}
for all \(\mQ\in\pmea{\outA}\) by \cite[{Lemma \ref*{C-lem:information}}]{nakiboglu19C}, as well.

The Augustin information is continuously differentiable in 
its order on \(\reals{+}\),
and its derivative is given by 
\begin{align}
\label{eq:lem:informationO:differentiability-alt}
\pder{}{\rno}\RMI{\rno}{\mP}{\Wm}
&=\begin{cases}
\tfrac{1}{(\rno-1)^2}\CRD{1}{\Wma{\rno}{\qmn{\rno,\mP}}}{\Wm}{\mP}
&\rno\neq 1
\\
\sum\nolimits_{\dinp}\tfrac{\mP(\dinp)}{2}
\VXS{\Wm(\dinp)}{\ln \der{\Wm(\dinp)}{\qmn{1,\mP}}}
&\rno= 1
\end{cases}
\end{align}
by \cite[Lemma {\ref*{C-lem:informationO}-(\ref*{C-informationO:differentiability})}]{nakiboglu19C},
where \(\Wma{\rno}{\qmn{\rno,\mP}}\) is the tilted channel defined as follows.
\begin{definition}\label{def:tiltedchannel}
	For any \(\rno\in\reals{+}\), \(\Wm:\inpS\to\pmea{\outA}\) and \(\mQ\in\pmea{\outA}\),
	\emph{the order \(\rno\) tilted channel \(\Wma{\rno}{\mQ}\)} is a function
	from  \(\{\dinp:\RD{\rno}{\Wm(\dinp)}{\mQ}<\infty\}\) to \(\pmea{\outA}\)
	given by
	\begin{align}
	\label{eq:def:tiltedchannel}
	\der{\Wma{\rno}{\mQ}(\dinp)}{\rfm}
	&\DEF e^{(1-\rno)\RD{\rno}{\Wm(\dinp)}{\mQ}}
	\left(\der{\Wm(\dinp)}{\rfm}\right)^{\rno} 
	\left(\der{\mQ}{\rfm}\right)^{1-\rno}.
	\end{align}
\end{definition}
The tilted channel can be used to express \(\RMI{\rno}{\mP}{\Wm}\)
in terms of the Kullback-Leibler divergences
using\footnote{It is worth noting that \eqref{eq:lem:information:alternative:opt} follows from 
	\eqref{eq:varitional-tilted} and \(\RMI{\rno}{\mP}{\Wm}=\CRD{\rno}{\Wm}{\qmn{\rno,\mP}}{\mP}\)
	for \(\rno\) values in \((0,1)\).}  
\cite[Lemma {\ref*{C-lem:information}}-({\ref*{C-information:alternative}})]{nakiboglu19C}:
\begin{align}
\label{eq:lem:information:alternative:opt}
\RMI{\rno}{\mP}{\Wm}
&=\tfrac{\rno}{1-\rno}\CRD{1}{\Wma{\rno}{\qmn{\rno,\mP}}}{\Wm}{\mP}
+\CRD{1}{\Wma{\rno}{\qmn{\rno,\mP}}}{\qmn{\rno,\mP}}{\mP}.
\end{align}
Furthermore, the Augustin mean satisfies
\begin{align}
\label{eq:augustinfixedpoint}
\sum\nolimits_{\dinp}\mP(\dinp) \Wma{\rno}{\qmn{\rno,\mP}}(\dinp)
&=\qmn{\rno,\mP}
\end{align}
and Augustin mean is the only probability measure 
satisfying both \(\qmn{1,\mP}\AC \mQ\)
and \(\sum\nolimits_{\dinp}\mP(\dinp) \Wma{\rno}{\mQ}(\dinp)=\mQ\)
by \cite[{Lemma \ref*{C-lem:information}}]{nakiboglu19C},
where \(\qmn{1,\mP}=\sum_{\dinp}\mP(\dinp)\Wm(\dinp)\). 
Thus for all \(\rno\in\reals{+}\) we have
\begin{align}
\label{eq:haroutunian}
\CRD{1}{\Wma{\rno}{\qmn{\rno,\mP}}}{\qmn{\rno,\mP}}{\mP}
&=\RMI{1}{\mP}{\Wma{\rno}{\qmn{\rno,\mP}}}.
\end{align}

\begin{definition}\label{def:capacity}
	For any \(\rno\in\reals{+}\), \(\Wm:\inpS\to\pmea{\outA}\), and \(\cset\subset\pdis{\inpS}\),
	\emph{the order \(\rno\) Augustin capacity of \(\Wm\) for the constraint set \(\cset\)} is 
	\begin{align}
	\label{eq:def:capacity}
	\CRC{\rno}{\Wm}{\cset}
	&\DEF \sup\nolimits_{\mP \in \cset}  \RMI{\rno}{\mP}{\Wm}.
	\end{align}
	When the constraint set \(\cset\) is the whole \(\pdis{\inpS}\), we denote the order \(\rno\) 
	Augustin capacity by \(\RC{\rno}{\Wm}\), i.e.,
	\(\RC{\rno}{\Wm}\DEF\CRC{\rno}{\Wm}{\pdis{\inpS}}\).
\end{definition}
Using the definitions of the Augustin information and capacity, 
we obtain  the following expression for the latter
\begin{align}
\label{eq:capacity}
\CRC{\rno}{\Wm}{\cset}
&=\sup\nolimits_{\mP \in \cset}\inf\nolimits_{\mQ\in\pmea{\outA}} \CRD{\rno}{\Wm}{\mQ}{\mP}.
\end{align}
If \(\cset\) is convex, then the order of the supremum and the infimum can be changed 
as a result of \cite[Thm. \ref*{C-thm:minimax}]{nakiboglu19C}:
\begin{align}
\label{eq:thm:minimax}
\sup\nolimits_{\mP \in \cset}\inf\nolimits_{\mQ\in\pmea{\outA}} \CRD{\rno}{\Wm}{\mQ}{\mP}
&=\inf\nolimits_{\mQ\in\pmea{\outA}}\sup\nolimits_{\mP \in \cset} \CRD{\rno}{\Wm}{\mQ}{\mP}.
\end{align}
If in addition \(\CRC{\rno}{\Wm}{\cset}\) is finite, then 
there exists a unique probability measure \(\qmn{\rno,\Wm,\cset}\),
\emph{called the order \(\rno\) Augustin center of \(\Wm\) for the constraint set \(\cset\)}, 
satisfying
\begin{align}
\label{eq:thm:minimaxcenter}
\CRC{\rno}{\Wm}{\cset}
&=\sup\nolimits_{\mP \in \cset} \CRD{\rno}{\Wm}{\qmn{\rno,\Wm,\cset}}{\mP}
\end{align}
by \cite[Thm. \ref*{C-thm:minimax}]{nakiboglu19C}.

We denote the set of all probability mass functions satisfying a cost constraint \(\costc\) by \(\cset(\costc)\), i.e.
\begin{align}
\label{eq:def:costconstraint}
\cset(\costc)
&\DEF \{\mP\in\pdis{\inpS}:\EXS{\mP}{\costf}\leq\costc\}.
\end{align}
With a slight abuse of notation, we denote
the cost-constrained Augustin capacity 
\(\CRC{\rno}{\Wm}{\cset(\costc)}\)
by
\(\CRC{\rno}{\Wm}{\costc}\),
as well.
A more detailed discussion of Augustin's information measures 
can be found in \cite{nakiboglu19C}. 

\subsection{The Sphere Packing Exponent}\label{sec:spherepackingexponent}
\begin{definition}\label{def:spherepackingexponent}
	For any \(\Wm:\inpS\to \pmea{\outA}\), \(\cset\subset\pdis{\inpS}\), 
	and \(\rate\in\reals{+}\), the sphere packing exponent (SPE)  is
	\begin{align}
	\label{eq:def:spherepackingexponent}
	\spe{\rate,\Wm,\cset}
	&\DEF \sup\nolimits_{\rno\in (0,1)} \tfrac{1-\rno}{\rno} \left(\CRC{\rno}{\Wm}{\cset}-\rate\right).
	\end{align}
	When the constraint set \(\cset\) is the whole \(\pdis{\inpS}\), 
	we denote SPE by \(\spe{\rate,\Wm}\), i.e.
	\(\spe{\rate,\Wm}\DEF\spe{\rate,\Wm,\pdis{\inpS}}\).
\end{definition}

With a slight abuse of notation, we denote 
SPE for \(\cset(\costc)\) by \(\spe{\rate,\Wm,\costc}\)
and SPE for \(\cset=\{\mP\}\)  
case by \(\spe{\rate,\Wm,\mP}\).
Note that as a result of definitions of Augustin capacity and 
the SPE we have
\begin{align}
\label{eq:lem:spherepacking:compositionconstrained}
\spe{\rate,\Wm,\cset}
&=\sup\nolimits_{\mP\in\cset} \spe{\rate,\Wm,\mP}.
\end{align}

Furthermore,
since  \(\RMI{\rno}{\mP}{\Wm}\) is continuously differentiable in \(\rno\)
by \cite[Lemma {\ref*{C-lem:informationO}-(\ref*{C-informationO:differentiability})}]{nakiboglu19C},
we can apply the derivative test to find the optimal \(\rno\) in 
\eqref{eq:def:spherepackingexponent} for \(\cset=\{\mP\}\) case:
using 
\eqref{eq:lem:informationO:differentiability-alt}
and
\eqref{eq:lem:information:alternative:opt}
we get
\begin{align}
\label{eq:parametric-derivative}
\pder{}{\rno}\tfrac{1-\rno}{\rno}\left(\RMI{\rno}{\mP}{\Wm}-\rate\right)
&=\tfrac{1}{\rno^{2}}\left(\rate-\CRD{1}{\Wma{\rno}{\qmn{\rno,\mP}}}{\qmn{\rno,\mP}}{\mP}\right).
\end{align}
On the other hand, either 
\(\CRD{1}{\Wma{\rno}{\qmn{\rno,\mP}}}{\qmn{\rno,\mP}}{\mP}=\RMI{1}{\mP}{\Wm}\) 
for all positive orders \(\rno\),
or \(\CRD{1}{\Wma{\rno}{\qmn{\rno,\mP}}}{\qmn{\rno,\mP}}{\mP}\) 
is increasing and continuous in the order \(\rno\) on \(\reals{+}\)
by 
\cite[Lemma {\ref*{C-lem:informationO}-(\ref*{C-informationO:monotonicityofharoutunianinformation})}]{nakiboglu19C}.
Furthermore,
\(\CRD{1}{\Wma{1}{\qmn{1,\mP}}}{\qmn{1,\mP}}{\mP}\) 
is equal to  \(\RMI{1}{\mP}{\Wm}\)
by definition and 
\(\lim\nolimits_{\rno \downarrow 0} \CRD{1}{\Wma{\rno}{\qmn{\rno,\mP}}}{\qmn{\rno,\mP}}{\mP}\)
is equal to \(\lim\nolimits_{\rno \downarrow 0}\RMI{\rno}{\mP}{\Wm}\) 
by \eqref{eq:haroutunian} and
\cite[Lemma {\ref*{C-lem:informationO}-(\ref*{C-informationO:limitofharoutunianinformation})}]{nakiboglu19C}.
Thus
for any rate \(\rate\) in 
\((\lim_{\rno \downarrow 0}\RMI{\rno}{\mP}{\Wm},\RMI{1}{\mP}{\Wm})\)
there exists an order \(\rns\in(0,1)\) satisfying 
\begin{align}
\label{eq:parametric-haroutunianform-rate}
\rate
&=\CRD{1}{\Wma{\rns}{\qmn{\rns,\mP}}}{\qmn{\rns,\mP}}{\mP}
\end{align}
by the intermediate value theorem \cite[4.23]{rudin}.
The order \(\rns\) satisfying \eqref{eq:parametric-haroutunianform-rate}
is unique  because \(\CRD{1}{\Wma{\rno}{\qmn{\rno,\mP}}}{\qmn{\rno,\mP}}{\mP}\)
is increasing in \(\rno\). 
The monotonicity of
\(\CRD{1}{\Wma{\rno}{\qmn{\rno,\mP}}}{\qmn{\rno,\mP}}{\mP}\) in \(\rno\)
and  \eqref{eq:parametric-derivative} also imply 
\(\spe{\rate,\Wm,\mP}=\tfrac{1-\rns}{\rns}\left(\RMI{\rns}{\mP}{\Wm}-\rate\right)\).
Thus as a result of  \eqref{eq:lem:information:alternative:opt}, 
the unique \(\rns\) satisfying \eqref{eq:parametric-haroutunianform-rate} also satisfies
\begin{align}
\label{eq:parametric-haroutunianform-exponent}
\spe{\rate,\Wm,\mP}
&=\CRD{1}{\Wma{\rns}{\qmn{\rns,\mP}}}{\Wm}{\mP}.
\end{align} 
Since \(\CRD{1}{\Wma{\rno}{\qmn{\rno,\mP}}}{\qmn{\rno,\mP}}{\mP}\) 
is continuous and increasing in \(\rno\), 
its inverse is increasing  and continuous, as well. 
Thus the definition of 
SPE given in \eqref{eq:def:spherepackingexponent} 
and the definition of derivative as a limit
imply that for any \(\rate\) in 
\((\lim_{\rno \downarrow 0}\RMI{\rno}{\mP}{\Wm},\RMI{1}{\mP}{\Wm})\)
the unique order \(\rns\)
satisfying \eqref{eq:parametric-haroutunianform-rate}
also satisfies
\begin{align}
\label{eq:parametric-haroutunianform-slope}
\pder{}{\rate}\spe{\rate,\Wm,\mP}
&=\tfrac{\rns-1}{\rns},
\end{align}
as was established in \cite[Lemma \ref*{D-lem:spherepacking-cc}]{nakiboglu18D}.

\section{The Refined Sphere Packing Bound}\label{sec:RSPB}
In this section, we consider the channel coding problem 
for various channel models
and derive lower bounds to the error probability of 
the following form
\begin{align}
\label{eq:rSPB}
\Pem{(\blx)}
&\geq A \blx^{\frac{\dspe{\rate}-1}{2}} e^{-\blx \spe{\rate}}
&
&\forall \blx \geq \blx_{0}.
\end{align}
for constants \(A\) and \(\blx_{0}\) determined by the rate, the channel,  
and the constraints on the codes ---if there exist any. 
Following \cite{altugW11,altugW14A,chengHT19}, 
we call these bounds refined sphere packing bounds (refined SPBs)
because of their resemblance to the standard SPBs,
e.g. \cite{shannonGB67A,haroutunian68,augustin69,
	ebert66,richters67,gallager,omura75,csiszarkorner,wyner88-b,
	burnashevK99,arikan02,lapidoth93,augustin78,nakiboglu19B,nakiboglu19E,nakiboglu17,nakiboglu18D,
	dalai13A,dalaiW17},
establishing\footnote{The approximation error terms in standard SPBs 
	are usually \(\bigo{\sqrt{\blx}}\) 
	or \(\bigo{\ln \blx}\), rather than just \(\smallo{\blx}\).}
\begin{align}
\label{eq:sSPB}
\Pem{(\blx)}
&\geq e^{-\blx\spe{\rate}-\smallo{\blx}}
&
&\forall \blx \geq \blx_{0}.
\end{align}
The refined SPBs that we state and prove in this section are not formally 
particular cases of a general proposition. Nevertheless, they are all consequences 
of Lemma \ref{lem:HTBE} and the properties of Augustin's information measures.

We establish a refined SPB for
the constant composition codes in Subsection \ref{sec:RSPB:constantcomposition},
for codes on (possibly) non-stationary \renyi symmetric channels 
in  Subsection \ref{sec:RSPB:symmetric},  
and 
for codes on additive white Gaussian noise channels with 
quadratic cost functions in  Subsection \ref{sec:RSPB-AWGNC}.

\subsection{Constant Composition Codes}\label{sec:RSPB:constantcomposition}
\begin{theorem}\label{thm:constantcomposition}
	For any \(\Wm:\inpS\to\pmea{\outA}\), 
	\(M,L,\blx\in\integers{+}\),
	\(\mP\in\pdis{\inpS}\) satisfying 
	\(\lim\nolimits_{\rno \downarrow 0}\RMI{\rno}{\mP}{\Wm}
	<\tfrac{1}{\blx}\ln \tfrac{M}{L}<\RMI{1}{\mP}{\Wm} \)
	and
	\(\blx\mP(\dinp)\in\integers{\geq0}\) for all \(\dinp\in\inpS\),
	the order \(\rns\DEF\tfrac{1}{1-\dspe{\frac{1}{\blx}\ln\frac{M}{L},\Wm,\mP}}\) 
	satisfies
	\begin{align}
	\label{eq:thm:constantcomposition-hypothesis}
	\CRD{1}{\Wma{\rns}{\qmn{\rns,\mP}}}{\qmn{\rns,\mP}}{\mP}
	&=\tfrac{1}{\blx}\ln \tfrac{M}{L}.
	\end{align}
	Furthermore, any \((M,L)\) channel code of length \(\blx\)
	whose codewords all have the same composition \(\mP\) satisfies
	\begin{align}
	\label{eq:thm:constantcomposition}
	\Pem{(\blx)}
	&\geq \htdelta^{\rns-1} (4\blx)^{-\frac{1}{2\rns}}
	e^{-\blx \spe{\frac{1}{\blx}\ln\frac{M}{L},\Wm,\mP}}	
	\end{align}	
	provided that 
	\(\sqrt{\amn{2}\blx}-\tfrac{\ln4\blx}{2\rns}\geq \ln \htdelta\)
	where
\begin{align}
\notag
\amn{2}&\DEF
\EXS{\mP\mtimes\Wma{\rns}{\qmn{\rns,\mP}}}{\abs{\ln\der{\Wm}{\qmn{\rns,\mP}}-\EXS{\Wma{\rns}{\qmn{\rns,\mP}}}{\ln\der{\Wm}{\qmn{\rns,\mP}}}}^{2}},
\\
\notag
\amn{3}&\DEF
\EXS{\mP\mtimes\Wma{\rns}{\qmn{\rns,\mP}}}{
	\abs{\ln\der{\Wm}{\qmn{\rns,\mP}}-\EXS{\Wma{\rns}{\qmn{\rns,\mP}}}{\ln\der{\Wm}{\qmn{\rns,\mP}}}
	}^{3}},
\\
\notag
\htdelta
&\DEF \exp\left(2\sqrt{2\pi e}\left(0.56\tfrac{\amn{3}}{\amn{2}}+\sqrt{\amn{2}}\right)\right).
\end{align}
\end{theorem}
Although Theorem \ref{thm:constantcomposition} itself is 
composition-dependent, 
it implies ---via appropriate worst-case assumptions--- 
composition-independent bounds, such as
\cite[Thm 1]{altugW14A}.
Similar composition-dependent \cite[Proposition 13]{chengHT19}
and composition independent \cite[Thm 8]{chengHT19} bounds
have, recently, been derived for classical-quantum channels using 
an approach similar to that of \cite{altugW14A}.
The primary advantages of Theorem \ref{thm:constantcomposition}
over the previous results are the conceptual simplicity and brevity 
of its proof and its definite approximation error terms.

\begin{proof}[Proof of Theorem \ref{thm:constantcomposition}]
	The existence of a unique order \(\rns\) satisfying
	\eqref{eq:thm:constantcomposition-hypothesis} was 
	proved, and its value was determined  in Section \ref{sec:preliminary}, 
	see \eqref{eq:parametric-haroutunianform-rate},
	\eqref{eq:parametric-haroutunianform-exponent},
	and
	\eqref{eq:parametric-haroutunianform-slope}.
	
	Let probability measures \(\wmn{\dmes}\), \(\mQ\),  and \(\vmn{\dmes}\)
	in \(\pmea{\outA_{1}^{\blx}}\) be
	\begin{align}
\notag
\wmn{\dmes} 
&\DEF \bigotimes\nolimits_{\tin=1}^{\blx} \Wm(\enc_{\tin}(\dmes)),
&
\mQ
&\DEF \bigotimes\nolimits_{\tin=1}^{\blx} \qmn{\rns,\mP},
&	
\vmn{\dmes} 
&\DEF \bigotimes\nolimits_{\tin=1}^{\blx} \Wma{\rns}{\qmn{\rns,\mP}}(\enc_{\tin}(\dmes)).
	\end{align}
	Then \(\vmn{\dmes}\) is equal to the order \(\rns\)
	tilted probability measure between \(\wmn{\dmes}\) and \(\mQ\).
	Furthermore, the empirical distribution of 
	the all of the codewords
	---i.e.,all \(\enc(\dmes)\)'s--- are equal to \(\mP\)
	by the hypothesis; thus we have	
	\begin{align}
	\notag
	\RD{1}{\vmn{\dmes}}{\mQ}
	&=\blx\CRD{1}{\Wma{\rns}{\qmn{\rns,\mP}}}{\qmn{\rns,\mP}}{\mP}
	&
	&\dmes\in\mesS, 
	\\
	\notag
	\RD{1}{\vmn{\dmes}}{\wmn{\dmes}}
	&=\blx\CRD{1}{\Wma{\rns}{\qmn{\rns,\mP}}}{\Wm}{\mP}
	&
	&\dmes\in\mesS. 
	\end{align}
On the other hand, \(\sum_{m\in\mesS}\mQ(\dmes\in\dec)\leq L\)
by the definition of list decoding. 
Thus 
at least half of the messages in \(\mesS\)
---at least \(\lfloor \tfrac{M+1}{2} \rfloor\) of the messages in \(\mesS\) to be precise---
 will satisfy \(\mQ(\dmes\in\dec)\leq \tfrac{2 L}{M}\)
as a result of Markov's inequality.
Applying Lemma \ref{lem:HTBE} 
with \(\oev=\{\dout_{1}^{\blx}:\dmes\in\dec(\dout_{1}^{\blx})\}\)
and \(\beta=2\) for the messages satisfying
\(\mQ(\dmes\in\dec)\leq \tfrac{2 L}{M}\), we get
\begin{align}
\notag
\Pem{\dmes}
&\geq\htdelta^{\rns-1}2^{\frac{\rns-1}{\rns}}
\blx^{-\frac{1}{2\rns}}e^{-\blx\CRD{1}{\Wma{\rns}{\qmn{\rns,\mP}}}{\Wm}{\mP}}	
\end{align}
as long as
\(\sqrt{\amn{2}\blx}-\tfrac{\ln4\blx}{2\rns}\geq \ln \htdelta\).
Then \eqref{eq:thm:constantcomposition}
follows from
\eqref{eq:parametric-haroutunianform-rate},
\eqref{eq:parametric-haroutunianform-exponent},
\eqref{eq:thm:constantcomposition-hypothesis},
and the definition error probability as the
average of the conditional error probabilities
of the messages.
\end{proof}

\subsection{Codes On \renyi Symmetric Channels}\label{sec:RSPB:symmetric}
\begin{definition}\label{def:symmetry}
A channel \(\Wm:\inpS\to\pmea{\outA}\) satisfying  
\(\Wm\AC\rfm\) for some \(\rfm\in\pmea{\outA}\)
is \renyi symmetric iff
for each \(\rno\in\reals{+}\) with finite \(\RC{\rno}{\Wm}\)
there exists a function \(\GX_{\rno}^{\Wm}(\cdot):\reals{}\to[0,1]\)
satisfying 
\begin{align}
\label{eq:def:symmetry}
\Wm\left(\left.
\left\{\der{\Wm(\dinp)}{\rfm}\leq e^{\mS}\der{\qmn{\rno,\Wm}}{\rfm}\right\} \right\vert\dinp\right)
&=\GX_{\rno}^{\Wm}(\mS)
&
&\forall\dinp\in\inpS,\mS\in\reals{}. 
\end{align}
\end{definition}

\begin{remark}
If \(\Wm\) is \renyi symmetric, then
the identity 
\(\lim\nolimits_{\mS\downarrow-\infty}\GX_{\rno}^{\Wm}(\mS)=0\)
holds 
whenever \(\RC{\rno}{\Wm}\) is finite.
On the other hand, the identity
\(\lim\nolimits_{\mS\uparrow\infty}\GX_{\rno}^{\Wm}(\mS)=1\)
is violated
whenever \(\Wm(\dinp)\NAC \qmn{\rno,\Wm}\),
which can only happen for \(\rno\)'s in \((0,1)\).
Such a \renyi symmetric \(\Wm\) is obtained
by removing \(\wma{}{\ind,\ind}\) from \({\cal W}_{\ind}\) described
in \cite[Example {\ref*{A-eg:singular-countable}}]{nakiboglu19A}
and the resulting \(\GX_{\rno}^{\Wm}\) is given by
\(\GX_{\rno}^{\Wm}(\mS)=(\sfrac{1}{2})\IND{\mS\geq\ln\!\sfrac{1}{2}}\).
\end{remark}

The \renyi symmetry holds 
for all input symmetric channels described in \cite[Definition 3.2]{wiechmanS08}
and 
for all the Gallager symmetric channels described in \cite[p. 94]{gallager}, 
see Appendix \ref{sec:renyi-symmetry}.
Recall that the Gallager symmetry holds for all
strongly symmetric (Dobrushin symmetric) channels, which is described in \cite{dobrushin62B}.
The binary symmetric channel is strongly symmetric.
The binary erasure channel is Gallager symmetric but not strongly symmetric.
The binary input Gaussian channel is input symmetric according 
but not Gallager symmetric.
The Rayleigh fading channel with per coherence interval
power constraint analyzed in \cite[(3)]{lanchoODKV19B} is 
\renyi symmetric, by \cite[(7) and (10)]{lanchoODKV19B},
but not input symmetric, see \cite[(5)]{lanchoODKV19B}.
\begin{remark}
The input symmetry described in \cite[Definition 3.2]{wiechmanS08} can be generalized 
by relying on a compact group with the associated Haar measure,
rather than a finite additive group with the uniform distribution.
The Rayleigh fading channel with per coherence interval
power constraint analyzed in \cite{lanchoODKV19B}
is input symmetric for this more general definition. 
The covariant channels analyzed by Holevo in \cite{holevo02}, 
can be seen as the counterparts of \cite[Definition 3.2]{wiechmanS08}
and its generalization in the framework of Quantum Information Theory.
\end{remark}

The derivation of the refined SPB for the \renyi symmetric channels
is analogous to the derivation of the refined SPB for the constant composition codes.
Lemma \ref{lem:parametric-symmetric}, given in the following,
is used  in lieu of
\eqref{eq:parametric-haroutunianform-rate},
\eqref{eq:parametric-haroutunianform-exponent},
\eqref{eq:parametric-haroutunianform-slope}
in the latter derivation.
\begin{lemma}\label{lem:parametric-symmetric}
For any \renyi symmetric channel \(\Wm:\inpS\to\pmea{\outA}\) 
with finite \(\RC{1}{\Wm}\)
and rate \(\rate\) in \((\lim_{\rno \downarrow 0}\RC{\rno}{\Wm},\RC{1}{\Wm})\)
there exists an order \(\rns\in(0,1)\) such that
	\begin{align}
	\label{eq:lem:parametric-symmetric-rate}
	\rate
	&=\RD{1}{\Wma{\rns}{\qmn{\rns,\Wm}}(\dinp)}{\qmn{\rns,\Wm}}
	&
	&\forall\dinp\in\inpS,
	\\
	\label{eq:lem:parametric-symmetric-exponent}
	\spe{\rate,\Wm}
	&=\RD{1}{\Wma{\rns}{\qmn{\rns,\Wm}}(\dinp)}{\Wm(\dinp)}
	&
	&\forall\dinp\in\inpS.
	\end{align}
Furthermore, if either 
\(\Wma{\rns}{\qmn{\rns,\Wm}}\left(\left.
\left\{\der{\Wm(\dinp)}{\rfm}=\gamma\der{\qmn{\rns,\Wm}}{\rfm}\right\}\right\vert\dinp\right)<1\)
for all \(\gamma\in\reals{+}\)
or \(\qmn{\rno,\Wm}=\mQ\) for all \(\rno\in(0,1]\), then
	\begin{align}
	\label{eq:lem:parametric-symmetric-slope}
	\pder{}{\rate}\spe{\rate,\Wm}
	&=\tfrac{\rns-1}{\rns}.
	\end{align}
\end{lemma}

Lemma \ref{lem:parametric-symmetric} is proved in Appendix \ref{sec:ProofOf:lem:parametric-symmetric}.
Proving the essential assertions of Lemma \ref{lem:parametric-symmetric} 
for input symmetric channels, however, is considerably easier:
for any input symmetric channel \(\Wm\)
and the uniform probability mass function \(\mU\) 
on its input set,
\(\RC{\rno}{\Wm}=\RMI{\rno}{\mU}{\Wm}\)
and \(\qmn{\rno,\Wm}=\qmn{\rno,\mU}\)
for all \(\rno\in\reals{+}\).
Consequently,
the identities given in 
\eqref{eq:lem:parametric-symmetric-rate},
\eqref{eq:lem:parametric-symmetric-exponent},
and
\eqref{eq:lem:parametric-symmetric-slope}
are nothing but 
the identities given in 
\eqref{eq:parametric-haroutunianform-rate},
\eqref{eq:parametric-haroutunianform-exponent},
and
\eqref{eq:parametric-haroutunianform-slope}
for \(\mP=\mU\) case
because the Kullback-Leibler divergences on the right-hand-sides of 
\eqref{eq:lem:parametric-symmetric-rate}
and
\eqref{eq:lem:parametric-symmetric-exponent}
have the same value for all \(\dinp\)
by the symmetry.
Hence, \eqref{eq:lem:parametric-symmetric-slope} holds
for any input symmetric channels
satisfying 
\(\lim_{\rno \downarrow 0}\RC{\rno}{\Wm}<\RC{1}{\Wm}\)
by \eqref{eq:parametric-haroutunianform-slope}, as well.

\begin{remark}\label{rem:inputsymmetric}
If \(\der{\Wm(\dinp)}{\qmn{\rno,\mU}}=\gamma\) holds \(\Wm(\dinp)\)-a.s. for all \(\dinp\)
for a \((\gamma,\rno)\) pair for an input symmetric \(\Wm\!\),
then  \(\qmn{\rno,\mU}=\sum_{\dinp} \mU(\dinp)\Wma{\rnt}{\qmn{\rno,\mU}}(\dinp)\)
for all \(\rnt\).
Thus \(\qmn{\rnt,\mU}=\qmn{\rno,\mU}\) 
and 
\(\RMI{\rnt}{\mU}{\Wm}=\ln\gamma\)
for all \(\rnt\in\reals{+}\)
by \cite[{Lemma \ref*{C-lem:information}}]{nakiboglu19C}. 
Thus such a \((\gamma,\rno)\) pair does not exists 
for input symmetric \(\Wm\)'s satisfying 
\(\lim_{\rno \downarrow 0}\RC{\rno}{\Wm}\neq\RC{1}{\Wm}\).
\end{remark}

\begin{remark}
Lemma \ref{lem:parametric-symmetric} is stated under the finite \(\RC{1}{\Wm}\) hypothesis,
yet it holds under the weaker hypothesis \(\lim\nolimits_{\rno\uparrow}\tfrac{1-\rno}{\rno}\RC{\rno}{\Wm}\), as well.
However, establishing Lemma \ref{lem:parametric-symmetric} under this weaker hypothesis would require 
us to introduce the concepts of power mean, \renyi information, and compactness in the topology of setwise convergence,
see \cite[Lemma  \ref*{A-lem:capacityEXT}-(\ref*{A-capacityEXT-compact-N})]{nakiboglu19A}.
\end{remark}

\begin{theorem}\label{thm:symmetric}
Let \(\Wmn{\tin}:\inpS_{\tin}\to\pmea{\outA_{\tin}}\) be
a \renyi symmetric channel with finite \(\RC{1}{\Wmn{\tin}}\)
for all \(\tin\in\integers{+}\)
and \(\blx,M,L\in\integers{+}\) satisfy
\(\lim_{\rno \downarrow 0}\RC{\rno}{\Wmn{[1,\blx]}}<\ln \tfrac{M}{L}<\RC{1}{\Wmn{[1,\blx]}}\).
Then there exists an order \(\rns\in(0,1)\) satisfying
\begin{align}
\label{eq:thm:symmetric-hypothesis}
\RD{1}{\Umn{\rns}(\dinp_{1}^{\blx})}{\qmn{\rns,\Wmn{[1,\blx]}}}
&=\ln\tfrac{M}{L}
&
&~
&
&\forall\dinp_{1}^{\blx}\in\inpS_{1}^{\blx}
\end{align}
where \(\Umn{\rno}\DEF\{\Wmn{[1,\blx]}\}_{\rno}^{\qmn{\rno,\Wmn{[1,\blx]}}}\)
for all \(\rno\in(0,1)\). 
Furthermore any \((M,L)\)  channel code 
on \(\Wmn{[1,\blx]}\) satisfies
\begin{align}
\label{eq:thm:symmetric}
\Pem{(\blx)}
&\geq \htdelta^{\rns-1} (4\blx)^{-\frac{1}{2\rns}}
e^{-\spe{\ln\frac{M}{L},\Wmn{[1,\blx]}}}	
\end{align}	
provided that \(\sqrt{\amn{2}\blx}-\tfrac{\ln4\blx}{2\rns}\geq \ln \htdelta\) where
\begin{align}
\notag
\amn{2}
&\DEF\tfrac{1}{\blx}\sum\nolimits_{\tin=1}^{\blx}
\EXS{\Umn{\rns}(\dinp_{1}^{\blx})}{\left(\cla{\rns}{\tin}(\dinp_{\tin})-
	\EXS{\Umn{\rns}(\dinp_{1}^{\blx})}{\cla{\rns}{\tin}(\dinp_{\tin})}\right)^{2}},
\\
\notag
\amn{3}
&\DEF\tfrac{1}{\blx}\sum\nolimits_{\tin=1}^{\blx}
\EXS{\Umn{\rns}(\dinp_{1}^{\blx})}{\abs{\cla{\rns}{\tin}(\dinp_{\tin})-
		\EXS{\Umn{\rns}(\dinp_{1}^{\blx})}{\cla{\rns}{\tin}(\dinp_{\tin})}}^{3}},
\\
\notag
\htdelta
&\DEF\exp\left(2\sqrt{2\pi e}\left(0.56\tfrac{\amn{3}}{\amn{2}}+\sqrt{\amn{2}}\right)\right),
\end{align}
and \(\cla{\rno}{\tin}(\dinp_{\tin})=\ln \der{\Wmn{\tin}(\dinp_{\tin})}{\rfm_{\tin}}-\ln \der{\qmn{\rno,\Wmn{\tin}}}{\rfm_{\tin}}\)
for all \(\rno\in(0,1)\).

Furthermore, if
\(\{\Wmn{\tin}\}_{\rns}^{\qmn{\rns,\Wmn{\tin}}}\left(\left.
\left\{\der{\Wmn{\tin}(\dinp_{\tin})}{\rfm_{\tin}}=\gamma\der{\qmn{\rns,\Wmn{\tin}}}{\rfm_{\tin}}\right\}\right\vert\dinp_{\tin}\right)<1\)
for all \(\gamma\in\reals{+}\) for some \(\tin\)
or \(\qmn{\rno,\Wmn{[1,\blx]}}=\mQ\) for all \(\rno\in(0,1]\),
then \(\rns=\tfrac{1}{1-\dspe{\ln\frac{M}{L},\Wmn{[1,\blx]}}}\).
\end{theorem}

Note that if any of  the component channels, i.e.,any of the \(\Wmn{\tin}\)'s, is an input symmetric 
channel satisfying 
\(\lim_{\rno \downarrow 0}\RC{\rno}{\Wmn{\tin}}<\RC{1}{\Wmn{\tin}}\),
then 
\(\rns=\tfrac{1}{1-\dspe{\ln\frac{M}{L},\Wmn{[1,\blx]}}}\)
holds as a result of Remark \ref{rem:inputsymmetric}.

Theorem \ref{thm:symmetric} does not assume the channel 
to be stationary,
i.e., it holds even when \(\Wmn{\tin}\)'s are not identical.
To the best of our knowledge, refined sphere packing bounds  
have only been reported for stationary channels before
---even in the case of symmetric channels considered in
\cite[(1.28)]{dobrushin62B}, 
\cite[Thm. 1]{altugW11}, \cite[Thm. 14]{chengHT19},
\cite[Corollary 2]{vazquezFKL18},
\cite[Thm. 4]{lanchoODKV19A},
\cite[(36), (37b)]{lanchoODKV19B},
\cite[Thm. 1]{altugW19}.

For the stationary input symmetric channels Theorem \ref{thm:symmetric}
is tight both in terms of exponent and prefactor for
rates above the critical rate, provided that
channel is not singular.
For the case of the singular stationary input symmetric channels,
\altug and Wagner \cite{altugW19} have recently reported a sharper result, 
which generalizes Elias's result in \cite{elias55B} 
for the binary erasure channels. 
In order to obtain such results, however,
merely plugging in bounds on binary hypothesis testing 
is not enough, see Section \ref{sec:conclusion} for a more detailed discussion.

\begin{remark}\label{rem:feedback}	
Theorem \ref{thm:symmetric} is derived using Lemma \ref{lem:HTBE}, 
which is stated for the product measures. 
Lemma \ref{lem:HTBE}, however, holds for any \(\mW\) and \(\mQ\) 
for which \(\cln{\tin}\)'s
are  independent random variables under \(\wma{\rno}{\mQ}\).
This condition is satisfied by 
the output distributions and 
the Augustin centers on the product channels with feedback,
i.e.,by \(\Wmn{\vec{[1,\blx]}}(\dinp)\) and \(\qmn{\rno,\Wmn{\vec{[1,\blx]}}}\),
provided that 
the component channels are \renyi symmetric.
Thus Theorem \ref{thm:symmetric} holds not just for codes on 
the product channel \(\Wmn{[1,\blx]}\)
but also for codes on the product channels with feedback 
\(\Wmn{\vec{[1,\blx]}}\).
Similar observations have been used to establish 
the SPB on product channels with feedback in
\cite{dobrushin62A,haroutunian77,augustin78,nakiboglu19B,altugW19}.
The formal definition of the product channels with feedback 
and the proof of the SPB on these channels without the symmetry assumptions
can be found in \cite{augustin78,nakiboglu19B,nakiboglu19E}.
\end{remark}

\begin{proof}[Proof of Theorem \ref{thm:symmetric}]
The \renyi symmetry and
\eqref{eq:thm:minimaxcenter} 
imply 
\begin{align}
\label{eq:symmetric-1}
\RC{\rno}{\Wmn{\tin}}
&=\RD{\rno}{\Wmn{\tin}(\dinp_{\tin})}{\qmn{\rno,\Wmn{\tin}}}
&
&
&
&\forall\dinp_{\tin}\in\inpS_{\tin},\tin\in\integers{+}.
\end{align}	
On the other hand, \cite[Lemma \ref*{C-lem:capacityproduct}]{nakiboglu19C} implies
\begin{align}
\label{eq:symmetric-2}
\RC{\rno}{\Wmn{[1,\blx]}}
&=\sum\nolimits_{\tin=1}^{\blx}\RC{\rno}{\Wmn{\tin}},
\\
\label{eq:symmetric-3}
\qmn{\rno,\Wmn{[1,\blx]}}
&=\bigotimes\nolimits_{\tin=1}^{\blx}\qmn{\rno,\Wmn{\tin}}.
\end{align}
Then the product structure 
\(\Wmn{[1,\blx]}(\dinp_{1}^{\blx})=\otimes_{\tin=1}^{\blx}\Wmn{\tin}(\dinp_{\tin})\)
and the \renyi symmetry of \(\Wmn{\tin}\)'s
imply the \renyi symmetry
of \(\Wmn{[1,\blx]}\). In particular
\begin{align}
\label{eq:symmetric-4}
\gX_{\rno}^{\Wmn{[1,\blx]}}
&=\gX_{\rno}^{\Wmn{1}}
\mtimes
\gX_{\rno}^{\Wmn{2}}
\mtimes
\cdots
\mtimes
\gX_{\rno}^{\Wmn{\blx}}
\end{align}
where  \(\mtimes\) denotes the convolution and 
\(\gX\)'s are the density functions of the corresponding \(\GX\)'s,
i.e., 
\(\gX\)'s and \(\GX\)'s are uniquely determined by each other via 
the following relation
\begin{align}
\notag
\GX_{\rno}^{\Wm}(\mS)
&=\int_{-\infty}^{\mS}\gX_{\rno}^{\Wm}(\dsta) \dif{\dsta}.
\end{align}
Since \(\Wmn{[1,\blx]}\) is \renyi symmetric the existence of 
an order \(\rns\in(0,1)\) satisfying \eqref{eq:thm:symmetric-hypothesis}
follows from \eqref{eq:lem:parametric-symmetric-rate}
of Lemma \ref{lem:parametric-symmetric}.

Let probability measures \(\wmn{\dmes}\), \(\mQ\),  and \(\vmn{\dmes}\)
in \(\pmea{\outA_{1}^{\blx}}\) be
\begin{align}
\notag
\wmn{\dmes} 
&\DEF \bigotimes\nolimits_{\tin=1}^{\blx} \Wmn{\tin}(\enc_{\tin}(\dmes)),
&
\mQ
&\DEF \bigotimes\nolimits_{\tin=1}^{\blx} \qmn{\rns,\Wmn{\tin}},
&
\vmn{\dmes} 
&\DEF \bigotimes\nolimits_{\tin=1}^{\blx} \{\Wmn{\tin}\}_{\rns}^{\qmn{\rns,\Wmn{\tin}}}(\enc_{\tin}(\dmes)).
\end{align}
Note that \(\wmn{\dmes}=\Wmn{[1,\blx]}(\enc(\dmes))\) by definition,
\(\mQ=\qmn{\rns,\Wmn{[1,\blx]}}\)
by \eqref{eq:symmetric-3},
and
\(\vmn{\dmes}\) is equal to the order \(\rns\)
tilted probability measure between \(\wmn{\dmes}\) and \(\mQ\), 
which is equal to \(\Umn{\rns}(\enc(\dmes))\),
by construction. 
Then Lemma \ref{lem:parametric-symmetric}
implies
\begin{align}
\notag
\RD{1}{\vmn{\dmes}}{\mQ}
&=\ln \tfrac{M}{L}
&
&\dmes\in\mesS, 
\\
\notag
\RD{1}{\vmn{\dmes}}{\wmn{\dmes}}
&=\spe{\ln\tfrac{M}{L},\Wmn{[1,\blx]}}
&
&\dmes\in\mesS. 
\end{align}
On the other hand, \(\sum_{m\in\mesS}\mQ(\dmes\in\dec)\leq L\)
by the definition of list decoding. 
Thus 
at least half of the messages in \(\mesS\) 
---at least \(\lfloor \tfrac{M+1}{2} \rfloor\) of the messages in \(\mesS\) to be precise---
will satisfy \(\mQ(\dmes\in\dec)\leq \tfrac{2 L}{M}\)
as a result of Markov's inequality.
Applying Lemma \ref{lem:HTBE} 
with \(\oev=\{\dout_{1}^{\blx}:\dmes\in\dec(\dout_{1}^{\blx})\}\)
and \(\beta=2\) for the messages satisfying
\(\mQ(\dmes\in\dec)\leq \tfrac{2 L}{M}\), we get
\begin{align}
\notag
\Pem{\dmes}
&\geq\htdelta^{\rns-1}2^{\frac{\rns-1}{\rns}}
\blx^{-\frac{1}{2\rns}}
e^{-\spe{\ln\frac{M}{L},\Wmn{[1,\blx]}}}	
\end{align}
provided that \(\sqrt{\amn{2}\blx}-\tfrac{\ln4\blx}{2\rns}\geq \ln \htdelta\).
Then \eqref{eq:thm:symmetric} follows from
the definition \(\Pe\)  as the
average of \(\Pem{\dmes}\).

Note that as a result of \eqref{eq:symmetric-4},
\(\gX_{\rns}^{\Wmn{[1,\blx]}}\) is a Dirac delta function 
iff all \(\gX_{\rns}^{\Wmn{\tin}}\)'s are.
This observation together with Lemma \ref{lem:parametric-symmetric}
implies the sufficient condition for \(\rns\) to be
\(\tfrac{1}{1-\dspe{\ln\frac{M}{L},\Wmn{[1,\blx]}}}\).
\end{proof}

\subsection{Codes On Additive White Gaussian Noise Channels}\label{sec:RSPB-AWGNC}
The additive white  Gaussian noise channel with noise variance \(\sigma^2\) 
is described via the following transition probability  
\begin{align}
\label{eq:GaussianChannel}
\Wm(\oev|\dinp)
&=\int_{\oev} \GausDen{\sigma^{2}}(\dout-\dinp) \dif{\dout}
&
&\forall \oev\in\rborel{\reals{}}
\end{align}
where \(\GausDen{\sigma^{2}}\) is the zero-mean Gaussian probability density function of variance \(\sigma^{2}\), i.e.,
\begin{align}
\notag
\GausDen{\sigma^{2}}(\dsta)
&=\tfrac{1}{\sqrt{2 \pi} \sigma} e^{-\frac{\dsta^{2}}{2\sigma^{2}}}
&
&\forall \dsta\in\reals{}.
\end{align}
With a slight abuse of notation,
we denote the corresponding probability measure
---i.e., zero-mean Gaussian probability measure of variance \(\sigma^{2}\)---
by \(\GausDen{\sigma^{2}}\), as well.

If the cost function is the quadratic one, 
then the zero-mean Gaussian distribution is the maximizer for 
the Augustin information among all input distributions
satisfying the cost constraint 
---for any positive order--- 
see \cite[Example \ref*{C-eg:SGauss}]{nakiboglu19C}, i.e.
\begin{align}
\notag
\CRC{\rno}{\Wm}{\costc}
&=\RMI{\rno}{\GausDen{\costc}}{\Wm}
\\
\notag
&=\sup\nolimits_{\mP:\EXS{\mP}{\dinp^{2}}\leq \costc} \RMI{\rno}{\mP}{\Wm}
&
&\forall \rno\in\reals{+}.
\end{align}
Furthermore, the order \(\rno\) Augustin center 
of this channel is a zero-mean Gaussian probability measure.
The closed-form expression for 
the Augustin  capacity and Augustin center 
are derived in 
\cite[(\ref*{C-eq:eg:SGauss-capacity}), (\ref*{C-eq:eg:SGauss-center}), (\ref*{C-eq:eg:SGauss-center-variance})]{nakiboglu19C}:
\begin{align}
\label{eq:gaussian:capacity}
\CRC{\rno}{\Wm}{\costc}
&=\begin{cases}
\tfrac{\rno \costc}{2(\rno \theta_{\rno}+(1-\rno)\sigma^{2})}
+
\frac{1}{\rno-1}
\ln\tfrac{\theta_{\rno}^{\sfrac{\rno}{2}}\sigma^{1-\rno}}{\sqrt{\rno\theta_{\rno}+(1-\rno)\sigma^{2}}}
&\rno\neq1
\\
\tfrac{1}{2}\ln \left(1+\tfrac{\costc}{\sigma^{2}}\right)
&\rno=1
\end{cases},
\\
\label{eq:gaussian:center}
\qmn{\rno,\Wm,\costc}
&=\GausDen{\theta_{\rno}},
\\
\label{eq:gaussian:center-variance}
\theta_{\rno}
&\DEF\sigma^{2}+\tfrac{\costc}{2}-\tfrac{\sigma^{2}}{2\rno}+\sqrt{(\tfrac{\costc}{2}-\tfrac{\sigma^{2}}{2\rno})^{2}+ \costc\sigma^2}.
\end{align}	
One can confirm the following identity by substitution 
\begin{align}
\notag
\rno\costc\theta_{\rno}
&=\rno(\theta_{\rno})^{2}- \left(2\rno-1\right)\theta_{\rno}\sigma^{2}+\left(\rno-1\right)\sigma^{4}
\\
\label{eq:gaussian:center-variance-necessarycondition}
&=\left(\theta_{\rno}-\sigma^{2}\right) 
\left(\rno\theta_{\rno}+(1-\rno)\sigma^{2}\right).
\end{align}
In fact, \(\theta_{\rno}\) is a root of the equality given above
because of a fixed point property similar to the one described in 
\eqref{eq:augustinfixedpoint}, 
see \cite[(\ref*{C-eq:eg:SGauss-Augustinoperator}), (\ref*{C-eq:eg:SGauss-necessarycondition})]{nakiboglu19C}
and the ensuing discussion.

The sphere packing exponent expression resulting from 
\eqref{eq:gaussian:capacity},
\eqref{eq:gaussian:center},
\eqref{eq:gaussian:center-variance},
and \eqref{eq:gaussian:center-variance-necessarycondition},
is derived in \cite[Example \ref{D-eg:SGauss}]{nakiboglu18D}. 
It is given by the following parametric form 
in  \cite[(\ref{D-eq:eg:SGauss-parametric-rate}), (\ref{D-eq:eg:SGauss-parametric-spe})]{nakiboglu18D}:
\begin{align}
\label{eq:gaussian:parametric-rate}
\rate
&=\tfrac{1}{2}\ln \tfrac{\rno \theta_{\rno}+(1-\rno) \sigma^{2}}{\sigma^{2}},
\\
\label{eq:gaussian:parametric-spe}
\hspace{-.15cm}
\spe{\rate,\Wm,\costc}
&=\tfrac{(1-\rno)\costc}{2(\rno\theta_{\rno}+(1-\rno)\sigma^{2})}
+\tfrac{1}{2}\ln \tfrac{\rno \theta_{\rno}+(1-\rno) \sigma^{2}}{\theta_{\rno}}.
\end{align}
Using \eqref{eq:gaussian:center-variance} and \eqref{eq:gaussian:parametric-rate}, 
one can express the unique \(\rns\) whose rate is \(\rate\) as a function of \(\rate\), as well:
\begin{align}
\label{eq:gaussian:parametric-rno}
\rns\DEF\tfrac{e^{2\rate}-1}{2}\left(\sqrt{1+\tfrac{4\sigma^2}{\costc}\tfrac{e^{2\rate}}{e^{2\rate}-1}}-1\right).
\end{align}
The equivalence of the parametric form given in 
\eqref{eq:gaussian:parametric-rate} and \eqref{eq:gaussian:parametric-spe}
to the expression  given by Gallager \cite[(7.4.33)]{gallager}  
can be confirmed by substitution using \eqref{eq:gaussian:parametric-rno}.
One can also confirm using
\eqref{eq:gaussian:center-variance} and \eqref{eq:gaussian:center-variance-necessarycondition}
in
\eqref{eq:gaussian:parametric-rate} and \eqref{eq:gaussian:parametric-spe} that
\begin{align}
\label{eq:gaussian:SPE-derivative}
\pder{}{\rate}\spe{\rate,\Wm,\costc}
&=\tfrac{\rns-1}{\rns}.
\end{align}
Furthermore, codes on additive white Gaussian noise channels satisfy 
both \eqref{eq:nonsingular} and \eqref{eq:rSPB}
as a result of Shannon's \cite[(3)]{shannon59};
this is established in Appendix \ref{sec:ShannonsExpressions}, 
for completeness. 
Theorem \ref{thm:gaussian}, which is establishing the refined SPB given in
\eqref{eq:rSPB}, is proved using principles and analysis similar 
to those used in the proofs of Theorems \ref{thm:constantcomposition} and \ref{thm:symmetric},
which are quite different from the ones employed in \cite{shannon59}.
\begin{theorem}\label{thm:gaussian}
Let \(\sigma\) and \(\costc\) be positive reals, 
\(\blx\), \(M\), and \(L\) be positive integers satisfying 
\(\rate\in
\left(0,\tfrac{1}{2}\ln\tfrac{\sigma^{2}+\costc}{\sigma^{2}}\right)\)
for \(\rate=\tfrac{1}{\blx}\ln\tfrac{M}{L}\),
and \(\Wmn{\tin}\)  be an additive white  Gaussian noise channel with 
noise variance \(\sigma^{2}\), say  \(\Wm\), for all \(\tin\).
Then any \((M,L)\) channel code \((\enc,\dec)\)
on \(\Wmn{[1,\blx]}\) satisfying
\(\sum_{\tin=1}^{\blx} (\enc_{\tin}(\dmes))^{2}=\blx\costc\)
for all messages \(\dmes\) in the message set \(\mesS\) 
satisfies
\begin{align}
\label{eq:thm:gaussian}
\Pem{(\blx)}
&\geq
\htdelta^{\rns-1} (4\blx)^{-\frac{1}{2\rns}}
e^{-\blx \spe{\rate,\Wm,\costc}}	
\end{align}	
for \(\rns\) given in \eqref{eq:gaussian:parametric-rno}
provided that 
\(\sqrt{\amn{2}\blx}-\tfrac{\ln \blx}{2\rns}\geq \htdelta\),
where 
\begin{align}
\notag
\amn{2}
&\DEF\tfrac{(\theta_{\rns}-\sigma^{2})^{2}}{2(\rns\theta_{\rns}+(1-\rns)\sigma^{2})^{2}}
+\tfrac{\sigma^{2}\theta_{\rns}\costc}{(\rns\theta_{\rns}+(1-\rns)\sigma^{2})^{3}},
\\
\notag
\amn{3}
&\DEF
\tfrac{18(\theta_{\rns}-\sigma^{2})^{3}}{(\rns\theta_{\rns}+(1-\rns)\sigma^{2})^{3}}
+\tfrac{18(\sqrt{\sfrac{2}{\pi}})\sigma^{3}\theta_{\rns}^{\sfrac{3}{2}}\costc^{\sfrac{3}{2}}}{(\rns\theta_{\rns}+(1-\rns)\sigma^{2})^{4.5}},
\\
\notag
\htdelta
&\DEF\exp\left(2\sqrt{2\pi e}\left(0.56\tfrac{\amn{3}}{\amn{2}}+\sqrt{\amn{2}}\right)\right).
\end{align}
\end{theorem}
Before proving of Theorem \ref{thm:gaussian}, let us briefly
discuss its implications.
Theorem \ref{thm:gaussian} bounds the performance of codes satisfying 
an equality cost constraint, but it can also be used to bound the performance 
of codes satisfying an inequality cost constraint.
In particular, Shannon has observed in \cite[(83)]{shannon59} that
\begin{align}
\label{eq:gaussian:Shannon}
\widetilde{\Pe}(\blx,\costc)
\geq 
\Pe(\blx,\costc)
\geq 
\widetilde{\Pe}(\blx+1,\costc),
\end{align}
where \(\Pe(\blx,\costc)\) is the infimum of 
error probabilities of \((M,L)\) channel codes 
satisfying \(\sum_{\tin=1}^{\blx}(\enc_{\tin}(\dmes))^{2}\leq\blx\costc\)
and \(\widetilde{\Pe}(\blx,\costc)\) is the analogous quantity for 
the constraint \(\sum_{\tin=1}^{\blx} (\enc_{\tin}(\dmes))^{2}=\blx\costc\).
The first inequality of \eqref{eq:gaussian:Shannon} holds 
because any code satisfying 
the equality constraint also satisfies the inequality constraint. 
The second inequality of \eqref{eq:gaussian:Shannon} 
is confirmed by considering 
an extension of codewords by one additional symbol, \(\enc_{\blx+1}(\dmes)\), 
so as to satisfy 
\(\sum_{\tin=1}^{\blx+1}(\enc_{\tin}(\dmes))^{2}=(\blx+1)\costc\).
Recently, Vazquez-Vilar have improved \eqref{eq:gaussian:Shannon}
in \cite[Proposition 1]{vazquez19B}
by observing that the same extension can be constructed for the constraint 
\(\sum_{\tin=1}^{\blx+1}(\enc_{\tin}(\dmes))^{2}=\blx \costc\), as well. 
Thus, we have
\begin{align}
\label{eq:gaussian:Vazquez-vilar}
\widetilde{\Pe}(\blx,\costc)
\geq 
\Pe(\blx,\costc)
\geq 
\widetilde{\Pe}\left(\blx+1,\tfrac{\blx}{\blx+1}\costc\right).
\end{align}
One can use Theorem \ref{thm:gaussian} together with either \eqref{eq:gaussian:Shannon} or \eqref{eq:gaussian:Vazquez-vilar}
to determine prefactor for codes satisfying the cost constraint with an inequality.
For that let us first note that \(\spe{\rate,\Wm,\costc}\) is convex in the rate \(\rate\)
as a result of \eqref{eq:gaussian:parametric-rno} and \eqref{eq:gaussian:SPE-derivative}
because \(\rns\) is increasing monotonically with \(\rate\) on \([0,\CRC{1}{\Wm}{\costc}]\).
Then \(\spe{\rate,\Wm,\costc}\) lies above its tangent at any point and 
\eqref{eq:gaussian:SPE-derivative} implies
\begin{align}
\label{eq:gaussian:convexity-expansion}
\spe{\rate_{1},\Wm,\costc}
&\geq\spe{\rate_{0},\Wm,\costc}+\tfrac{\rns(\rate_{0})-1}{\rns(\rate_{0})}(\rate_{1}-\rate_{0}). 
\end{align}
Applying Theorem \ref{thm:gaussian} at \((\blx+1)\) ---rather than \(\blx\)---
together with \eqref{eq:gaussian:Shannon} and
invoking \eqref{eq:gaussian:convexity-expansion}
for \(\rate_{0}=\tfrac{1}{\blx+1}\ln\tfrac{M}{L}\) and \(\rate_{1}=\rate\),
we get the following bound 
for \((M,L)\) channel codes \((\enc,\dec)\) satisfying
\(\sum_{\tin=1}^{\blx} (\enc_{\tin}(\dmes))^{2}\leq\blx\costc\)
\begin{align}
\notag
\Pem{(\blx)}
&\geq
\tfrac{\htdelta^{\rns_{0}-1}}{(4(\blx+1))^{\sfrac{1}{2\rns_{0}}}}
e^{-(\blx+1) \spe{\rate_{0},\Wm,\costc}}	
\\
\notag
&\geq
\tfrac{\htdelta^{\rns_{0}-1} e^{-\spe{\rate,\Wm,\costc}}}{(8\blx)^{\sfrac{1}{2\rns_{0}}}}
e^{-\blx\spe{\rate,\Wm,\costc}}e^{\frac{\rns_{0}-1}{\rns_{0}}\frac{\blx \rate}{\blx+1}}	
\\
\notag
&=\tfrac{\htdelta^{\rns_{0}-1}e^{\CRC{\rns_{0}}{\Wm}{\costc}\sfrac{(1-\rns_{0})}{\rns_{0}}}}{8^{\sfrac{1}{2\rns_{0}}}
	\blx^{\sfrac{(\rns-\rns_{0})}{2\rns\rns_{0}}}
	}
\blx^{-\frac{1}{2\rns}}
e^{-\blx\spe{\rate,\Wm,\costc}}
\end{align}	
where \(\rns_{0}=\rns(\rate_{0})\) and \(\amn{2}\), \(\amn{3}\), and \(\htdelta\) are 
calculated at \(\rns_{0}\), rather than \(\rns\),
provided that 
\(\sqrt{\amn{2}(\blx+1)}-\tfrac{\ln(\blx+1)}{2\rns_{0}}\geq \htdelta\).
Note that 
\(\abs{\rate-\rate_{0}}=\tfrac{\rate}{\blx+1}\)
and the function
\(\rns\) given in \eqref{eq:gaussian:parametric-rno}
is analytical in the rate \(\rate\);
hence \(\abs{\rns_{0}-\rns}\) is \(\bigo{\blx^{-1}}\).
Thus 
\begin{align}
\notag
\tfrac{\htdelta^{\rns_{0}-1}e^{\CRC{\rns_{0}}{\Wm}{\costc}\sfrac{(1-\rns_{0})}{\rns_{0}}}}{8^{\sfrac{1}{2\rns_{0}}}
	\blx^{\sfrac{(\rns-\rns_{0})}{2\rns\rns_{0}}}}
\sim A
\end{align}
for some constant \(A\in\reals{+}\) and
\eqref{eq:rSPB},
holds not only for codes satisfying 
\(\sum_{\tin=1}^{\blx} (\enc_{\tin}(\dmes))^{2}=\blx\costc\)
but also for codes satisfying 
\(\sum_{\tin=1}^{\blx} (\enc_{\tin}(\dmes))^{2}\leq\blx\costc\),
for the order \(\rns\) given in \eqref{eq:gaussian:parametric-rno}, 
for some constant \(A\in\reals{+}\) for \(\blx\) large enough.

\begin{proof}[Proof of Theorem \ref{thm:gaussian}]
	The following expressions for the order one \renyi divergence and 
	the order \(\rno\) tilted channel for the zero-mean Gaussian distribution 
	of variance \(\theta\) can be confirmed by substitution 
	\begin{align}
	\label{eq:gaussian-1}
	\RD{1}{\Wm(\dinp)}{\GausDen{\theta}}
	&=\tfrac{\sigma^{2}+\dinp^{2}-\theta}{2\theta}+\tfrac{1}{2}\ln \tfrac{\theta}{\sigma^{2}}
	\\
	\label{eq:gaussian-2}
	\Wma{\rno}{\GausDen{\theta}}(\oev|\dinp)
	&=\int_{\oev} \GausDen{\frac{\sigma^{2}\theta}{\rno\theta+(1-\rno)\sigma^{2}}}
	(\dout-\tfrac{\rno \theta}{\rno\theta+(1-\rno)\sigma^{2}}\dinp)\dif{\dout}
	\end{align}
	for all
	\(\dinp\in\reals{}\),
	\(\theta\in\reals{+}\),
	and  \(\oev\in\rborel{\reals{}}\).
Then as a result of  
\eqref{eq:gaussian:center-variance-necessarycondition},
\eqref{eq:gaussian:parametric-rate},
\eqref{eq:gaussian:parametric-spe},
and  
\eqref{eq:gaussian:parametric-rno},
we have
\begin{align}
\label{eq:gaussian-3}
\hspace{-.35cm}\CRD{1}{\Wma{\rns}{\GausDen{\theta_{\rns}}}}{\GausDen{\theta_{\rns}}}{\mP}
&=\tfrac{{\rns}^{2}\theta_{\rns}(\EXS{\mP}{\dinp^{2}}-\costc)}{2(\rns \theta_{\rns}+(1-\rns)\sigma^{2})^{2}}+\rate,
\\
\label{eq:gaussian-4}
\hspace{-.35cm}
\CRD{1}{\Wma{\rns}{\GausDen{\theta_{\rns}}}}{\Wm}{\mP}
&=\tfrac{(1-\rns)^{2}\sigma^{2}(\EXS{\mP}{\dinp^{2}}-\costc)}{2(\rns \theta_{\rns}+(1-\rns)\sigma^{2})^{2}}
+\spe{\rate,\Wm,\costc},
\end{align} 
where \(\rate=\tfrac{1}{\blx}\ln\tfrac{M}{L}\).

Let probability measures \(\wmn{\dmes}\), \(\mQ\),  and \(\vmn{\dmes}\)
in \(\pmea{\outA_{1}^{\blx}}\) be
\begin{align}
\notag
\wmn{\dmes} 
&\DEF \bigotimes\nolimits_{\tin=1}^{\blx} \Wm(\enc_{\tin}(\dmes)),
&
\mQ
&\DEF \bigotimes\nolimits_{\tin=1}^{\blx} \GausDen{\theta_{\rns}},
&\vmn{\dmes} 
&\DEF \bigotimes\nolimits_{\tin=1}^{\blx} \Wma{\rns}{\GausDen{\theta_{\rns}}}(\enc_{\tin}(\dmes)).
\end{align}
Then \(\vmn{\dmes}\) is the order \(\rns\)
tilted probability measure between \(\wmn{\dmes}\) and \(\mQ\);
using the hypothesis 
\(\sum_{\tin=1}^{\blx} (\enc_{\tin}(\dmes))^{2}=\blx\costc\)
together with \eqref{eq:gaussian-3} and \eqref{eq:gaussian-4},
we get
\begin{align}
\label{eq:gaussian-8}
\RD{1}{\vmn{\dmes}}{\mQ}
&=\ln \tfrac{M}{L}
&
&\dmes\in\mesS, 
\\
\label{eq:gaussian-9}
\RD{1}{\vmn{\dmes}}{\wmn{\dmes}}
&=\blx \spe{\tfrac{1}{\blx}\ln\tfrac{M}{L},\Wm,\costc}
&
&\dmes\in\mesS. 
\end{align}
In order to obtain \eqref{eq:thm:gaussian}, we will apply Lemma \ref{lem:HTBE}
to the probability measure pairs \((\wmn{\dmes},\mQ)\)  satisfying 
\(\mQ(\dmes\in\dec)\leq \tfrac{2 L}{M}\) for \(\beta=2\)
together with a rotation on \(\reals{}^{\blx}\) 
that minimizes the approximation error terms arising from the absolute 
third order moments.

For any \(\dinp\in\reals{}\), let the random variable \(\cln{\dinp}\) be 
\begin{align}
\notag
\cln{\dinp}
&\DEF \ln\der{\Wm(\dinp)}{\GausDen{\theta_{\rns}}}
-\EXS{\Wma{\rns}{\GausDen{\theta_{\rns}}}(\dinp)}{\ln\der{\Wm(\dinp)}{\GausDen{\theta_{\rns}}}}.
\end{align}
Then using \eqref{eq:gaussian-2} we get
\begin{align}
\notag
\cln{\dinp}
&=\tfrac{1}{2}
\tfrac{\dout^{2}}{\theta_{\rns}}
-\tfrac{(\dout -\dinp)^{2}}{2\sigma^{2}} 
+\tfrac{1}{2}\ln \tfrac{\theta_{\rns}}{\sigma^{2}}
-\tfrac{1}{2}\EXS{\Wma{\rns}{\GausDen{\theta_{\rns}}}(\dinp)}{\tfrac{\dout^{2}}{\theta_{\rns}}-\tfrac{(\dout -\dinp)^{2}}{\sigma^{2}}}
-\tfrac{1}{2}\ln \tfrac{\theta_{\rns}}{\sigma^{2}}
\\
\label{eq:gaussian-11}
&=\tfrac{\sigma^{2}-\theta_{\rns}}{2\sigma^{2}\theta_{\rns}}
\left[(\dout-\tfrac{\rns \theta_{\rns} \dinp}{\rns\theta_{\rns}+(1-\rns)\sigma^{2}})^{2}
-\tfrac{\sigma^{2} \theta_{\rns}}{\rns\theta_{\rns}+(1-\rns)\sigma^{2}}
\right]
+\tfrac{\dinp}{\rns\theta_{\rns}+(1-\rns)\sigma^{2}}
(\dout-\tfrac{\rns \theta_{\rns} \dinp}{\rns\theta_{\rns}+(1-\rns)\sigma^{2}}).
\end{align}
On the other hand, moments and absolute moments of a zero-mean Gaussian random variable \(\sta\) 
with variance \(\sigma^{2}\) satisfy the following identities
\begin{align}
\notag
\EX{\sta^{\knd}}
&=\IND{\frac{\knd}{2}\in\integers{+}}\sigma^{\knd} (\knd-1)!!,
\\
\notag
\EX{\abs{\sta}^{\knd}}
&=\left[\IND{\frac{\knd}{2}\notin\integers{+}}\sqrt{\tfrac{2}{\pi}}
+\IND{\frac{\knd}{2}\in\integers{+}}
\right]
\sigma^{\knd} (\knd-1)!!,
\end{align}
where \(\knd!!=\prod_{\ind=0}^{\lceil\frac{\knd}{2} \rceil-1}(\knd-2\ind)\).
Furthermore, for any three random variables 
\(\sta_{1}\), \(\sta_{2}\), and \(\sta_{3}\), 
we have\footnote{The inequality given in \eqref{eq:bound-on-the-third-moment}
	follows from the H\"{o}lder's inequality via the observation that
	the geometric mean is less than the arithmetic mean. 
	A proof is presented in Appendix \ref{sec:ProofOfeq:bound-on-the-third-moment}
	for completeness.}
\begin{align}
\label{eq:bound-on-the-third-moment}
\EX{\abs{\sta_{1}+\sta_{2}+\sta_{3}}^{3}}
&\leq9\left(
\EX{\abs{\sta_{1}}^{3}}+\EX{\abs{\sta_{2}}^{3}}+\EX{\abs{\sta_{3}}^{3}}
\right).
\end{align}
Then using \eqref{eq:gaussian-11}, one can confirm by substitution that
\begin{align}
\label{eq:gaussian-15}
\EXS{\Wma{\rns}{\GausDen{\theta_{\rns}}}(\dinp)}{(\cln{\dinp})^{2}}
&=\tfrac{(\theta_{\rns}-\sigma^{2})^{2}}{2(\rns\theta_{\rns}+(1-\rns)\sigma^{2})^{2}}
+\tfrac{\sigma^{2}\theta_{\rns}\dinp^{2}}{(\rns\theta_{\rns}+(1-\rns)\sigma^{2})^{3}},
\\
\label{eq:gaussian-16}
\EXS{\Wma{\rns}{\GausDen{\theta_{\rns}}}(\dinp)}{\abs{\cln{\dinp}}^{3}}
&\leq
\tfrac{18(\theta_{\rns}-\sigma^{2})^{3}}{(\rns\theta_{\rns}+(1-\rns)\sigma^{2})^{3}}
+\tfrac{18(\sqrt{\sfrac{2}{\pi}})\sigma^{3}\theta_{\rns}^{\sfrac{3}{2}}\dinp^{3}}{(\rns\theta_{\rns}+(1-\rns)\sigma^{2})^{4.5}}.
\end{align}
The hypothesis \(\sum_{\tin=1}^{\blx} (\enc_{\tin}(\dmes))^{2}=\blx\costc\)
implies the \(\amn{2}\) of Lemma \ref{lem:HTBE}  to be equal to 
the \(\amn{2}\) of Theorem \ref{thm:gaussian}.
One, however, cannot assert the analogous relation for \(\amn{3}\)'s.
Nevertheless, there exists a rotation in \(\reals{}^{\blx}\), say \(\alg{S}_{\dmes}\)
such that \(\alg{S}_{\dmes}\enc(\dmes)\) is equal to vector whose all entries are \(\sqrt{\costc}\), 
i.e.
\begin{align}
\notag
\alg{S}_{\dmes}(\enc_{1}(\dmes),\ldots,\enc_{\blx}(\dmes))
&=(\sqrt{\costc},\ldots,\sqrt{\costc}).
\end{align}
Note that for  \(\wmn{*}\DEF\bigotimes\nolimits_{\tin=1}^{\blx} \Wm(\sqrt{\costc})\),
we have
\begin{align}
\notag
\wmn{\dmes}(\dmes\notin \dec)
&=\wmn{*}(\dmes\notin \alg{S}_{\dmes}\dec),
\\
\notag
\mQ(\dmes\in \dec)
&=\mQ(\dmes\in \alg{S}_{\dmes}\dec).
\end{align}
Thus one can apply Lemma \ref{lem:HTBE} to the pair \((\wmn{*},\mQ)\) in
order to bound \(\Pem{\dmes}\). For the pair \((\wmn{*},\mQ)\), 
however, \(\amn{3}\) of Lemma \ref{lem:HTBE} 
is equal to the \(\amn{3}\) of of Theorem \ref{thm:gaussian}.

As it was the case for the proofs of Theorems \ref{thm:constantcomposition} and
\ref{thm:symmetric}, \(\sum_{m\in\mesS}\mQ(\dmes\in\dec)\leq L\)
by the definition of list decoding. 
Thus at least half of the messages in \(\mesS\) 
will satisfy \(\mQ(\dmes\in\dec)\leq \tfrac{2 L}{M}\)
as a result of Markov's inequality.
Applying Lemma \ref{lem:HTBE} 
with \(\oev=\{\dout_{1}^{\blx}:\dmes\in\dec(\dout_{1}^{\blx})\}\)
and \(\beta=2\) for the messages satisfying
\(\mQ(\dmes\in\dec)\leq \tfrac{2 L}{M}\)
and using \eqref{eq:gaussian-8} and \eqref{eq:gaussian-9}, we get
\begin{align}
\notag
\Pem{\dmes}
&\geq\htdelta^{\rns-1}2^{\frac{\rns-1}{\rns}}
\blx^{-\frac{1}{2\rns}}e^{-\blx \spe{\frac{1}{\blx}\ln\frac{M}{L},\Wm,\costc}}	
\end{align}
as long as
\(\sqrt{\amn{2}\blx}-\tfrac{\ln4\blx}{2\rns}\geq \ln \htdelta\).
Then \eqref{eq:thm:gaussian}
follows from the definition of \(\Pe\) as the
average of the conditional error probabilities.
\end{proof}

\section{Discussion}\label{sec:conclusion}
Theorems \ref{thm:constantcomposition}, \ref{thm:symmetric}, \ref{thm:gaussian} 
establish refined sphere packing bounds, i.e.,bounds of the form \eqref{eq:rSPB}, 
for fixed composition codes on stationary memoryless channels,
codes on (possibly) non-stationary \renyi symmetric channels,
and cost-constrained codes on additive white Gaussian noise channels with the quadratic cost function.
Derivations of Theorems \ref{thm:constantcomposition}, \ref{thm:symmetric}, \ref{thm:gaussian} 
rely on the properties of Augustin's information measures
and the application of Berry-Esseen theorem to the hypothesis testing problem
summarized in Lemma \ref{lem:HTBE}.
For certain cases including 
the additive white Gaussian channels \cite{shannon59}
and 
the strongly symmetric channels \cite{dobrushin62B},
these bounds are known to be tight in the sense that 
they can be matched by achievability results asserting 
the existence of codes satisfying
\begin{align}
\label{eq:rRC}
\Pem{(\blx)}
&\leq \widetilde{A} \blx^{\frac{\dspe{\rate}-1}{2}} e^{-\blx \spe{\rate}}
&
&\forall \blx \geq \blx_{0}
\end{align}
for rates between the critical rate and the capacity of the channel.
Recently, \altug and Wagner \cite[Theorem 1]{altugW19} 
have generalized the results of 
\cite{elias55B} and \cite{dobrushin62B} and
established \eqref{eq:rSPB} and \eqref{eq:rRC} for all
non-singular Gallager symmetric channels.

At least since \cite{elias55B}, it is also known 
that for the binary erasure channel 
the polynomial prefactor of \eqref{eq:rSPB} can be improved 
from \(\blx^{\sfrac{(\dspe{\rate}-1)}{2}}\) to 
\(\blx^{\sfrac{-1}{2}}\).
Recently, \altug and Wagner proved 
this result for all singular Gallager symmetric channels,
\cite[Theorem 2]{altugW19}.
Both \cite{elias55B} and \cite{altugW19},
however, have refrained from relying on bounds on 
the performance of the binary hypothesis testing problem 
with independent samples. 
This is not surprising because 
Lemma \ref{lem:HTBE} characterizes the 
prefactor for the binary hypothesis testing problem 
with independent samples, exactly.
Thus the refined SPBs of the form \eqref{eq:rSPB}
are the best possible bounds
for derivations of the SPB relying on the asymptotic
behavior of sums of independent random variables,
notwithstanding their suboptimality for 
singular Gallager symmetric channels.

As pointed out in Section \ref{sec:HypothesisTesting},
one can improve Lemma \ref{lem:HTBE} and determine not only
the prefactor  but also the asymptotic constant in 
the tradeoff between the probabilities of type I and type II errors,
either by invoking finer characterizations of the asymptotic behavior 
of sums of independent random
variables  or by applying a saddle point approximation.
Although such results, e.g. \cite{csiszarL71,vazquezFKL18}, 
require stronger hypothesis and are more 
nuanced,\footnote{Even the statement of these results are more nuanced 
	because	they need to distinguish the lattice and non-lattice cases 
	for the random variables involved in the analysis.}
they are important in the context of binary hypothesis testing.
From the standpoint of the channel coding problem, however, 
it is rather hard to justify the extra effort such an analysis requires.
First of all, the corresponding achievability results will have different 
constants, even when the prefactors match,
as observed in \cite{elias55B,shannon59,dobrushin62B}.
More importantly, such refined results on binary hypothesis testing
will also suffer from the subtlety discussed in the previous 
paragraph for the case of the singular channels,
i.e.,their prefactor will be suboptimal for the singular 
channels, such as the binary erasure channel.

The principal novelty of this manuscript is the use of 
the Berry-Essen theorem via suitable Augustin information measures 
to bound the optimal error probability in the channel coding problem.
In this manuscript, our primary focus was the rates below 
the channel capacity; 
thus, we have derived refined sphere packing bounds. 
The same idea can be used to strengthen the strong converse bounds 
under similar symmetry hypothesis, as it has recently been demonstrated
in \cite{chengN20A}.
The essential technical challenge in this line of work is the 
derivation of the refined SPBs and the refined strong converses 
without any symmetry assumptions on the channel or on the codes.
\appendices
\section{\renyi Symmetry is implied by input symmetry and Gallager Symmetry}\label{sec:renyi-symmetry}
In the following, we will explain briefly why the \renyi symmetry holds both 
for all input symmetric channels described in \cite[Definition 3.2]{wiechmanS08}
and for all Gallager symmetric channels described in \cite[p. 94]{gallager}.
Let us start with the input symmetric channels.
Let \(\mU\) be the uniform distribution 
on the input set of the input symmetric \(\Wm\)
and let \(\qmn{\rno}\) be
\begin{align}
\label{eq:def:sym:center}
\qmn{\rno}
&\DEF\tfrac{\mmn{\rno}}{\lon{\mmn{\rno}}}
&&&
&\mbox{~where~}
&&&
\der{\mmn{\rno}}{\rfm}
&\DEF
\left[\sum\nolimits_{\dinp} \mU(\dinp)\left(\der{\Wm(\dinp)}{\rfm}\right)^{\rno}\right]^{\frac{1}{\rno}}.
\end{align}
Then the using the input symmetry one can confirm that
\begin{align}
\label{eq:def:sym:divergence}
\RD{\rno}{\Wm(\dinp)}{\qmn{\rno}}
&=\tfrac{\rno}{\rno-1}\ln \lon{\mmn{\rno}}
&
&\forall \dinp\in\inpS.
\end{align}
The definitions of the tilted channel, \eqref{eq:def:sym:center},
and \eqref{eq:def:sym:divergence} imply 
\(\qmn{\rno}=\sum_{\dinp} \mU(\dinp)\Wma{\rno}{\qmn{\rno}}(\dinp)\)
and  \(\qmn{1,\mU}\AC\qmn{\rno}\). 
Thus \(\qmn{\rno}\) is the order \(\rno\) Augustin mean for the uniform 
input distribution 
and 
\(\RMI{\rno}{\mU}{\Wm}=\CRD{\rno}{\Wm}{\qmn{\rno}}{\mU}\)
by \cite[{Lemma \ref*{C-lem:information}}]{nakiboglu19C}. 
Hence
\begin{align}
\label{eq:def:sym:information}
~~~~\RMI{\rno}{\mU}{\Wm}
&=\tfrac{\rno}{\rno-1}\ln \lon{\mmn{\rno}}.
&
&\end{align}
Since \(\CRD{\rno}{\Wm}{\qmn{\rno}}{\mP}=\RMI{\rno}{\mU}{\Wm}\)
for all \(\mP\in\pdis{\inpS}\) by \eqref{eq:def:sym:divergence}
and \eqref{eq:def:sym:information},
the probability measure \(\qmn{\rno}\) is not only the order \(\rno\)
Augustin mean for the input distribution \(\mU\) 
but also the order \(\rno\) Augustin center of \(\Wm\) by 
\cite[Thm. \ref*{C-thm:minimax}]{nakiboglu19C}.
Then the constraint for the \renyi symmetry follows from  
the definition of \(\qmn{\rno}\) and the input symmetry. 
Thus every input symmetric channel \(\Wm\) is also a
\renyi symmetric channel.

Now let us considers a Gallager symmetric channel \(\Wm\).
Let \(\set{S}_{1},\ldots,\set{S}_{m}\) be the partition of the output set \(\outS\) 
assumed in the definition of Gallager symmetry, e.g., \cite[p. 94]{gallager},
\(\mU\) be the uniform distribution on the input set \(\inpS\) of \(\Wm\), 
and
\(\mmn{\rno}\) and \(\qmn{\rno}\) be the measures defined in \eqref{eq:def:sym:center}. 
Gallager symmetry  implies not only that \(\mmn{\rno}\) and \(\qmn{\rno}\)
are probability mass functions but also that they satisfy the following identities:
\begin{align}
\label{eq:def:gsym:center}
\mmn{\rno}(\dout)
&=\sum\nolimits_{\ind}\IND{\dout\in\set{S}_{\ind}}\abs{\set{S}_{\ind}}^{-\frac{1}{\rno}}\left(\sum\nolimits_{\mS\in\set{S}_{\ind}} 
[\Wm(\mS|\dinp)]^{\rno} \right)^{\frac{1}{\rno}} 
&
&\forall \dinp\in\inpS,\dout\in\outS.
\end{align}
Note that \(\mmn{\rno}(\dout)=\mmn{\rno}(\dsta)\), and hence \(\qmn{\rno}(\dout)=\qmn{\rno}(\dsta)\), 
whenever \(\dout\) and \(\dsta\) are in the same \(\set{S}_{\ind}\)
as a result of \eqref{eq:def:gsym:center}.
Using this fact together with the Gallager symmetry, 
one can confirm both \eqref{eq:def:sym:divergence}
and \(\qmn{\rno}=\sum_{\dinp} \mU(\dinp)\Wma{\rno}{\qmn{\rno}}(\dinp)\).
On the other hand, \(\qmn{1,\mU}\AC\qmn{\rno}\).
Thus \(\qmn{\rno}\) is the order \(\rno\) Augustin mean for the uniform 
input distribution 
and 
\(\RMI{\rno}{\mU}{\Wm}=\CRD{\rno}{\Wm}{\qmn{\rno}}{\mU}\)
by \cite[{Lemma \ref*{C-lem:information}}]{nakiboglu19C}. 
Hence \eqref{eq:def:sym:information} holds for Gallager symmetric \(\Wm\!\),
as well.
Thus \(\CRD{\rno}{\Wm}{\qmn{\rno}}{\mP}=\RMI{\rno}{\mU}{\Wm}\)
for all \(\mP\in\pdis{\inpS}\) 
and  \(\qmn{\rno}\) is not only the order \(\rno\)
Augustin mean for the input distribution \(\mU\) 
but also the order \(\rno\) Augustin center of \(\Wm\) by 
\cite[Thm. \ref*{C-thm:minimax}]{nakiboglu19C}.
Then the constraint for the \renyi symmetry follows from  
the definition of \(\qmn{\rno}\) given in \eqref{eq:def:sym:center}, 
the Gallager symmetry and \eqref{eq:def:gsym:center}. 
Thus every Gallager symmetric channel \(\Wm\) is 
also a \renyi symmetric channel.

\section{Proof of Lemma \ref{lem:parametric-symmetric}}\label{sec:ProofOf:lem:parametric-symmetric}
We first prove the existence of an order \(\rns\) in \((0,1)\) satisfying 
\eqref{eq:lem:parametric-symmetric-rate} using the intermediate value theorem.
The \renyi symmetry of \(\Wm\) implies
\begin{align}
\notag
\RD{\rno}{\Wm(\dinp)}{\qmn{\rnt,\Wm}}
&=\sup\nolimits_{\dsta\in\inpS} \RD{\rno}{\Wm(\dsta)}{\qmn{\rnt,\Wm}}.
\end{align}
for all \(\dinp\in\inpS\), \(\rnt\in(0,1]\), and 
\(\rno\in\reals{+}\).
Then \eqref{eq:capacity}, \eqref{eq:thm:minimax}, 
and	\eqref{eq:thm:minimaxcenter} imply 
\begin{align}
\label{eq:parametric-symmetric-1}
\RC{\rno}{\Wm}
&=\inf\nolimits_{\rnt\in(0,1]} \RD{\rno}{\Wm(\dinp)}{\qmn{\rnt,\Wm}}
\\
\label{eq:parametric-symmetric-2}
&=\RD{\rno}{\Wm(\dinp)}{\qmn{\rno,\Wm}}
\end{align}	
for all \(\dinp\in\inpS\)
and \(\rno\in(0,1]\).
Thus the non-negativity of the \renyi divergence 
and \eqref{eq:varitional-tilted} imply
\begin{align}
\label{eq:parametric-symmetric-3-q}
\RD{1}{\Wma{\rno}{\qmn{\rno,\Wm}}(\dinp)}{\qmn{\rno,\Wm}}
&\leq  \RC{\rno}{\Wm}
\\
\label{eq:parametric-symmetric-3-w}
\RD{1}{\Wma{\rno}{\qmn{\rno,\Wm}}(\dinp)}{\Wm(\dinp)}
&\leq
\tfrac{1-\rno}{\rno}\RC{\rno}{\Wm}
\end{align}
for all \(\dinp\in\inpS\) and \(\rno\in(0,1)\).
On the other hand the Pinsker's inequality
imply for all \(\dinp\in\inpS\) and \(\rno\in(0,1)\)
\begin{align}
\notag
\lon{\Wma{\rno}{\qmn{\rno,\Wm}}(\dinp)-\Wm(\dinp)}
&\leq \sqrt{2\RD{1}{\Wma{\rno}{\qmn{\rno,\Wm}}(\dinp)}{\Wm(\dinp)}}.
\end{align}
Thus \(\lim\nolimits_{\rno \uparrow 1} \lon{\Wma{\rno}{\qmn{\rno,\Wm}}(\dinp)-\Wm(\dinp)}=0\)
for all \(\dinp\in\inpS\) as a result of \eqref{eq:parametric-symmetric-3-w}.
On the other hand, the Augustin center \(\qmn{\rno,\Wm}\)
is continuous in \(\rno\) on \((0,1]\) for the total 
variation topology on \(\pmea{\outA}\) by 
\cite[Lemmas \ref*{C-lem:capacityO}-(\ref*{C-capacityO-continuity})
and \ref*{C-lem:centercontinuity}]{nakiboglu19C}.
Then the lower semicontinuity of the \renyi divergence 
in its arguments  for the topology of setwise convergence,
i.e.,\cite[Thm. 15]{ervenH14}, implies
\begin{align}
\notag
\liminf\nolimits_{\rno \uparrow 1} 
\RD{1}{\Wma{\rno}{\qmn{\rno,\Wm}}(\dinp)}{\qmn{\rno,\Wm}}
&\geq \RD{1}{\Wm(\dinp)}{\qmn{1,\Wm}}.
\end{align}
Then using  \eqref{eq:parametric-symmetric-2} we get
\begin{align}
\label{eq:parametric-symmetric:limitatone}
\liminf\nolimits_{\rno \uparrow 1} \RD{1}{\Wma{\rno}{\qmn{\rno,\Wm}}(\dinp)}{\qmn{\rno,\Wm}}
&\geq \RC{1}{\Wm}.
\end{align}
Furthermore, \(\RD{1}{\Wma{\rno}{\qmn{\rno,\Wm}}(\dinp)}{\qmn{\rno,\Wm}}\)
is continuous in \(\rno\) on \((0,1)\)
by \cite[Lemma \ref*{B-lem:tilting}]{nakiboglu19B}
because \(\qmn{\rno,\Wm}\)
is continuous in \(\rno\) on \((0,1]\) for the total 
variation topology.
Then \eqref{eq:parametric-symmetric-3-q} implies 
\begin{align}
\label{eq:parametric-symmetric:limitatzero}
\limsup\nolimits_{\rno \downarrow 0} \RD{1}{\Wma{\rno}{\qmn{\rno,\Wm}}(\dinp)}{\qmn{\rno,\Wm}}
&\leq
\lim\nolimits_{\rno \downarrow 0} \RC{\rno}{\Wm}.
\end{align}
The existence of an \(\rns\in(0,1)\) satisfying 
\eqref{eq:lem:parametric-symmetric-rate} follows from
\eqref{eq:parametric-symmetric:limitatone}, \eqref{eq:parametric-symmetric:limitatzero}, 
and the continuity of \(\RD{1}{\Wma{\rno}{\qmn{\rno,\Wm}}(\dinp)}{\qmn{\rno,\Wm}}\)
in \(\rno\) on \((0,1)\)
by the intermediate value theorem \cite[4.23]{rudin}.

We proceed with showing that any order \(\rns\) satisfying 
\eqref{eq:lem:parametric-symmetric-rate} 
also satisfies \eqref{eq:lem:parametric-symmetric-exponent}.
The definition of the SPE given 
in \eqref{eq:def:spherepackingexponent}
and
the consequence of the \renyi symmetry given in  \eqref{eq:parametric-symmetric-1},
imply that
\begin{align}
\notag
\spe{\rate,\Wm}
&=
\sup\nolimits_{\rno\in(0,1)} \inf\nolimits_{\rnt\in(0,1]}\tfrac{1-\rno}{\rno}
\left[\RD{\rno}{\Wm(\dinp)}{\qmn{\rnt,\Wm}}-\rate\right]
\\
\label{eq:parametric-symmetric-4}
&\leq \sup\nolimits_{\rno\in(0,1)} \fX(\rno,\rate),
\end{align}
where \(\fX(\rno,\tau)\DEF\tfrac{(1-\rno)}{\rno}[\RD{\rno}{\Wm(\dinp)}{\qmn{\rns,\Wm}} -\tau]\).
We show in the following that
\begin{align}
\label{eq:parametric-symmetric-5}
\sup\nolimits_{\rno\in(0,1)}\fX(\rno,\rate)
&=\RD{1}{\Wma{\rns}{\qmn{\rns,\Wm}}(\dinp)}{\Wm(\dinp)}.
\end{align}
Then for any order \(\rns\) satisfying \eqref{eq:lem:parametric-symmetric-rate},
equations \eqref{eq:varitional-tilted}, \eqref{eq:parametric-symmetric-2}
and the definition of SPE given in \eqref{eq:def:spherepackingexponent} imply
\begin{align}
\notag
\sup\nolimits_{\rno\in(0,1)}\fX(\rno,\rate)
&=\tfrac{(1-\rns)}{\rns}[\RC{\rns}{\Wm}-\rate]
\\
\label{eq:parametric-symmetric-6}
&\leq \spe{\rate,\Wm}.
\end{align}
Thus the inequalities given in 
\eqref{eq:parametric-symmetric-4} and \eqref{eq:parametric-symmetric-6}
hold as equalities and \(\rns\) satisfying \eqref{eq:lem:parametric-symmetric-rate}
also satisfies \eqref{eq:lem:parametric-symmetric-exponent} by
\eqref{eq:parametric-symmetric-5}.

Now we prove the identity given in \eqref{eq:parametric-symmetric-5},
which we have assumed in the preceding.
The \renyi divergence is non-decreasing in its order
by \cite[Thm. 3]{ervenH14}
and \(\RD{\rno}{\Wm(\dinp)}{\qmn{\rns,\Wm}}=\tfrac{\rno}{1-\rno}\RD{1-\rno}{\qmn{\rns,\Wm}}{\Wm(\dinp)}\) 
for all \(\rno\) in \((0,1)\) by definition. 
Then \(\RD{\rno}{\Wm(\dinp)}{\qmn{\rns,\Wm}}\) is finite for all 
\(\rno\) in \((0,1)\).
Thus both
\(\RD{\rno}{\Wm(\dinp)}{\qmn{\rns,\Wm}}\) 
and
\(\RD{1}{\Wma{\rno}{\qmn{\rns,\Wm}}(\dinp)}{\qmn{\rns,\Wm}}\)
are continuously differentiable in \(\rno\) on \((0,1)\) 
by  \cite[Lemma \ref{C-lem:analyticity}]{nakiboglu19C}
and their derivatives on \((0,1)\) are 
\begin{align}
\label{eq:parametric-symmetric-7}
\der{}{\rno}\RD{\rno}{\Wm(\dinp)}{\qmn{\rns,\Wm}}
&=\tfrac{1}{(\rno-1)^{2}}\RD{1}{\Wma{\rno}{\qmn{\rns,\Wm}}(\dinp)}{\Wm(\dinp)},
\\
\label{eq:parametric-symmetric-9}
\der{}{\rno}\RD{1}{\Wma{\rno}{\qmn{\rns,\Wm}}(\dinp)}{\qmn{\rns,\Wm}}
&=\rno \VXS{\Wma{\rno}{\qmn{\rns,\Wm}}(\dinp)}{\cln{\dinp}},
\end{align}
where \(\cln{\dinp}=\ln\der{\Wm(\dinp)}{\rfm}-\ln\der{\qmn{\rns,\Wm}}{\rfm}\). Then
\begin{align}
\label{eq:parametric-symmetric-10}
\pder{}{\rno}\fX(\rno,\tau)
&=\tfrac{1}{\rno^{2}}\left(\tau-\RD{1}{\Wma{\rno}{\qmn{\rns,\Wm}}(\dinp)}{\qmn{\rns,\Wm}}\right).
\end{align}
Consequently, 
\eqref{eq:lem:parametric-symmetric-rate} implies \eqref{eq:parametric-symmetric-5}
by the derivative test
provided that 
\(\cln{\dinp}=\gamma\) does not hold \(\Wma{\rno}{\qmn{\rns,\Wm}}(\dinp)\)-a.s.
for any \(\gamma\in\reals{+}\) and \(\rno\in(0,1)\).
If \(\Wma{\rno}{\qmn{\rns,\Wm}}(\{\cln{\dinp}=\gamma\}|\dinp)=1\)
for some \(\gamma\in\reals{+}\) and \(\rno\in(0,1)\), on the other hand, then
the identities \eqref{eq:parametric-symmetric-11}, \eqref{eq:parametric-symmetric-12},
and \eqref{eq:parametric-symmetric-13} hold  for all \(\rno\in(0,1)\)
and one can confirm \eqref{eq:parametric-symmetric-5} by substitution.
\begin{align}
\label{eq:parametric-symmetric-11}
\RD{\rno}{\Wm(\dinp)}{\qmn{\rns,\Wm}}
&=\gamma+\tfrac{1}{1-\rno}
\ln\tfrac{1}{\Wm(\{\cln{\dinp}=\ln\gamma\}|\dinp)},
\\
\label{eq:parametric-symmetric-12}
\RD{1}{\Wma{\rno}{\qmn{\rns,\Wm}}(\dinp)}{\qmn{\rns,\Wm}}
&=\gamma+\ln\tfrac{1}{\Wm(\{\cln{\dinp}=\ln\gamma\}|\dinp)},
\\
\label{eq:parametric-symmetric-13}
\RD{1}{\Wma{\rno}{\qmn{\rns,\Wm}}(\dinp)}{\Wm(\dinp)}
&=\ln\tfrac{1}{\Wm(\{\cln{\dinp}=\ln\gamma\}|\dinp)}.
\end{align}

Now we are left with establishing \eqref{eq:lem:parametric-symmetric-slope} with 
either of the additional hypotheses. 
Let us first assume that there does not exists a \(\gamma\)	satisfying
\(\Wma{\rns}{\qmn{\rns,\Wm}}\left(\left.
\left\{\der{\Wm(\dinp)}{\rfm}=\gamma\der{\qmn{\rns,\Wm}}{\rfm}\right\}\right\vert\dinp\right)=1\).
Then \(\VXS{\Wma{\rno}{\qmn{\rns,\Wm}}(\dinp)}{\cln{\dinp}}>0\)
for all \(\rno\) in \((0,1)\)
and thus \(\RD{1}{\Wma{\rno}{\qmn{\rns,\Wm}}(\dinp)}{\qmn{\rns,\Wm}}\)
is increasing in \(\rno\) on \((0,1)\).
Since \(\RD{1}{\Wma{\rno}{\qmn{\rns,\Wm}}(\dinp)}{\qmn{\rns,\Wm}}\) 
is also continuous in \(\rno\) on \((0,1)\), 
it has a continuous and increasing inverse function.
Then as a result of \eqref{eq:parametric-symmetric-10}
there exists an \(\epsilon>0\) and an increasing continuous function 
\(\hX:(\rate-\epsilon,\rate+\epsilon)\to(0,1)\) 
satisfying
\begin{align}
\notag
\sup\nolimits_{\rno\in(0,1)} \fX(\rno,\tau)
&=\fX(\hX(\tau),\tau)
&
&\forall \tau\in(\rate-\epsilon,\rate+\epsilon).
\end{align}
Then \eqref{eq:parametric-symmetric-4} implies
\begin{align}
\label{eq:parametric-symmetric-14}
\spe{\tau,\Wm}
&\leq \fX(\hX(\tau),\tau)
&
&\forall \tau\in(\rate-\epsilon,\rate+\epsilon).
\end{align}
On the other hand,
\eqref{eq:def:spherepackingexponent} and \eqref{eq:parametric-symmetric-2} 
imply
\begin{align}
\label{eq:parametric-symmetric-15}
\spe{\tau,\Wm}
&\geq \fX(\rns,\tau)
&
&\forall \tau\in\reals{+}.
\end{align}
Furthermore, \eqref{eq:parametric-symmetric-6} implies
\begin{align}
\label{eq:parametric-symmetric-16}
\spe{\rate,\Wm}
&\geq \fX(\rno,\rate)
&
&\forall \rno\in(0,1).
\end{align}
Then \eqref{eq:lem:parametric-symmetric-slope} follows from
\eqref{eq:parametric-symmetric-14}, \eqref{eq:parametric-symmetric-15}, \eqref{eq:parametric-symmetric-16},
the identity \(\hX(\rate)=\rns\),
the definition of the derivative as a limit,
and the continuity of \(\hX(\cdot)\),  which is defined as the inverse function of 
\(\RD{1}{\Wma{\rno}{\qmn{\rns,\Wm}}(\dinp)}{\qmn{\rns,\Wm}}\)
as a function of \(\rno\).

Now we establish \eqref{eq:lem:parametric-symmetric-slope} assuming
the existence of a \(\mQ\) satisfying 
\(\qmn{\rno,\Wm}=\mQ\)  for all \(\rno\) in \((0,1)\).
If there does not exists a \(\gamma\) satisfying
\(\Wma{\rns}{\mQ}\left(\left.
\left\{\der{\Wm(\dinp)}{\rfm}=\gamma\der{\mQ}{\rfm}\right\}\right\vert\dinp\right)=1\),
then \eqref{eq:lem:parametric-symmetric-slope} holds by the preceding discussion.
If there exists such a \(\gamma\) , then as a result of 
\eqref{eq:parametric-symmetric-2} and \eqref{eq:parametric-symmetric-11} we have
\begin{align}
\label{eq:parametric-symmetric-17}
\RC{\rno}{\Wm}
&=\gamma+\tfrac{1}{1-\rno}
\ln\tfrac{1}{\Wm(\{\cln{\dinp}=\ln\gamma\}|\dinp)}
&
&\forall \rno\in(0,1)
\end{align}
Then such a \(\gamma\) does not exist by
the hypotheses of the lemma because 
\eqref{eq:parametric-symmetric-17} 
and \(\RC{1}{\Wm}=\lim_{\rno \uparrow 1}\RC{\rno}{\Wm}\)
imply \(\RC{1}{\Wm}=\infty\) for 
\(\Wm(\{\cln{\dinp}=\ln\gamma\}|\dinp)<1\) case and 
\(\lim_{\rno \downarrow 0}\RC{\rno}{\Wm}=\RC{1}{\Wm}\)
for \(\Wm(\{\cln{\dinp}=\ln\gamma\}|\dinp)=1\)
case.

\section{Shannon's Bounds For AWGN Channels and The Sphere Packing Exponent}\label{sec:ShannonsExpressions}
Shannon, \cite[(3)]{shannon59}, bounded the error probability of length \(\blx\) block codes 
described in Theorem \ref{thm:gaussian} as
\begin{align}
\label{eq:ShannonsBounds}
Q(\theta)
&\leq \Pem{(\blx)}
\leq
Q(\theta)-\int_{0}^{\theta} \tfrac{\Omega(\xi)}{\Omega(\theta)}
\dif{Q(\xi)}
\end{align}
where
\(\Omega(\cdot):[0,\pi]\to[0,\tfrac{2\pi^{\blx/2}}{\Gamma(\blx/2)}]\)
is the function mapping the cone angle 
to the corresponding solid angle in \(\reals{}^{\blx}\),
\(\theta\) is the cone angle satisfying
\(\Omega(\theta)=\tfrac{\Omega(\pi)}{M}\),
and
\(Q(\xi)\) is the probability that a point \(X\) in 
\(\reals{}^{\blx}\)
at a distance \(\sqrt{\blx\costc}\)
from the origin \(O\) being moved outside a circular cone of half-angle
\(\xi\) with the vertex at the origin \(O\) and the axis at \(OX\)
by a Gaussian noise of variance \(\sigma^{2}\).

Shannon, \cite[(4) and (5)]{shannon59}, derived 
the exact asymptotic expressions for both 
the upper bound and the lower bound given in \eqref{eq:ShannonsBounds}
in terms of functions \(f(\cdot)\) and \(g(\cdot)\) that do not depend on
the block length \(\blx\):
\begin{align}
\label{eq:ShannonLowerBound-EA}
Q(\theta)
&\sim\tfrac{f(\theta)}{\sqrt{\blx}} e^{-\blx \SGEX(\theta)}
&
&\forall \theta>\theta_{c}, 
\\
\label{eq:ShannonUpperBound-EA}
Q(\theta)-\int_{0}^{\theta}\tfrac{\Omega(\xi)}{\Omega(\theta)}
\dif{Q(\xi)}
&\sim\tfrac{g(\theta)}{\sqrt{\blx}} e^{-\blx \SGEX(\theta)}
&
&\forall \theta\in(\theta_{c},\theta_{cr}),
\end{align}
where \(\theta_{c}\) and \(\theta_{cr}\) ---the cone angles corresponding to 
the channel capacity and the critical rate---
are given by 
\begin{align}
\label{eq:ShannonCap}
\theta_{c}
&\DEF\arcsin\left[ \left(1+\tfrac{\costc}{\sigma^{2}}\right)^{-\sfrac{1}{2}}\right],
\\
\label{eq:ShannonCR}
\theta_{cr}
&\DEF\arcsin\left[\left(\tfrac{1}{2}+\tfrac{\costc}{4\sigma^{2}}+\sqrt{\tfrac{1}{4}+(\tfrac{\costc}{4\sigma^{2}})^{2}}\right)^{-\sfrac{1}{2}}\right],
\end{align}
and \(\SGEX(\cdot)\) ---the fixed cone angle exponent--- is defined via the function \(\GX(\cdot)\) as follows
\begin{align}
\label{eq:ShannonExp}
\SGEX(\theta)
&\DEF\tfrac{\costc}{2\sigma^{2}}
-\tfrac{\sqrt{\costc}}{2\sigma}G(\theta)\cos\theta-\ln\left(G(\theta)\sin\theta\right),
\\
\label{eq:ShannonGF}
\GX(\theta)
&\DEF\tfrac{1}{2}\left(\tfrac{\sqrt{\costc}}{\sigma}\cos\theta+\sqrt{\tfrac{\costc}{\sigma^{2}}\cos^{2}\theta+4} \right).
\end{align}
\begin{remark}
	Shannon's notation in \cite{shannon59} is slightly different from ours; Shannon works with 
	the signal to noise ``amplitude'' ratio \(A\DEF\tfrac{\sqrt{\costc}}{\sigma}\), 
	rather than cost constraint \(\costc\) and the noise power \(\sigma^{2}\).
	Furthermore, Shannon specifies the critical cone angle \(\theta_{cr}\) 
	as the solution of the equation given in \eqref{eq:ShannonCR-equation}.
	Nevertheless, one obtains the closed-form expression given in \eqref{eq:ShannonCR},
	by plugging in the definition of \(\GX(\cdot)\) ---given in \eqref{eq:ShannonGF}---
	in \eqref{eq:ShannonCR-equation} and solving the resulting quadratic 
	equation for \(\sin^{2}\theta_{cr}\). 
	\begin{align}
	\label{eq:ShannonCR-equation}
	2\cos\theta_{cr}-\tfrac{\sqrt{\costc}}{\sigma}\GX(\theta_{cr})\sin^{2}\theta_{cr}=0.
	\end{align} 
\end{remark}
Shannon presented the exact asymptotic expression for the rate \(\rate\) in terms of 
the cone angle \(\theta\) in \cite[(11)]{shannon59}, as well:
\begin{align}
\label{eq:ShannonRate}
e^{\blx\rate}
&\sim\tfrac{\sqrt{2\pi \blx} \cos\theta}{\sin^{\blx-1} \theta}. 
\end{align}
We obtain the fixed-rate asymptotic expression 
corresponding to the fixed cone angle asymptotic expressions 
\eqref{eq:ShannonLowerBound-EA}, \eqref{eq:ShannonUpperBound-EA},
and \eqref{eq:ShannonExp}, 
by first deriving the asymptotic expression for the cone angle 
\(\theta\) for a fixed-rate \(\rate\) using \eqref{eq:ShannonRate}.
If \(\abs{\delta_{\blx}}\ll 1\) and
\begin{align}
\label{eq:smallangle}
\theta_{\blx}=\xi+\delta_{\blx},
\end{align}
then using  the small angle approximation for  the trigonometric functions
we get
\begin{align}
\notag
\tfrac{\sqrt{2\pi \blx}  \cos\theta_{\blx}}{\sin^{\blx-1} \theta_{\blx}}
&=\tfrac{\sqrt{2\pi \blx}\cos\xi}{\sin^{\blx-1}\xi} 
\tfrac{1-\delta_{\blx}\tan\xi+\bigo{\delta_{\blx}^{2}}}{\left(1+\delta_{\blx}\cot\xi+\bigo{\delta_{\blx}^{2}}\right)^{\blx-1}}.
\end{align}
Invoking \(\ln(1+\epsilon)=\epsilon+\bigo{\epsilon^{2}}\), we get
\begin{align}
\notag
\ln\tfrac{\sqrt{2\pi \blx} \cos\theta_{\blx}}{\sin^{\blx-1} \theta_{\blx}}
&=\blx\ln\tfrac{1}{\sin\xi}+
\ln \left(\sqrt{\tfrac{\pi \blx}{2}}\sin 2\xi \right)
-\left[(\blx-1)\cot\xi+\tan\xi\right]\delta_{\blx}
+\blx\bigo{\delta_{\blx}^{2}}
\end{align}
Consequently, if 
\begin{align}
\label{eq:phi}
\xi
&=\arcsin(e^{-\rate}),
\\
\label{eq:delta}
\delta_{\blx}
&=\tfrac{1}{\blx\cot\xi}\ln \left(\sqrt{\tfrac{\pi \blx}{2}}\sin 2\xi\right),
\end{align}
then the rate corresponding to the cone angle \(\theta_{\blx}\) 
at the block length \(\blx\) is \(\rate+\bigo{\tfrac{\ln^{2}\blx}{\blx^{2}}}\).
In other words, we get a fixed-rate by changing the cone angle by an additive factor proportional 
to \(\tfrac{\ln \blx}{\blx}\).
In order to obtain the exact asymptotic expressions for the 
upper and lower bounds to the error probability given in \eqref{eq:ShannonsBounds}
via \eqref{eq:ShannonLowerBound-EA} and \eqref{eq:ShannonUpperBound-EA} at 
a fixed-rate \(\rate\), we will apply Taylor's expansion. 
To that end, we first calculate the derivatives of \(\GX\) and
\(\SGEX\).
As a result of  \eqref{eq:ShannonGF}, we have
\begin{align}
\notag
\der{}{\theta}\GX
&=\tfrac{1}{2}\left(-\tfrac{\sqrt{\costc}}{\sigma}\sin\theta-\tfrac{(\sfrac{\costc}{\sigma^{2}})\sin\theta \cos\theta}{\sqrt{(\sfrac{\costc}{\sigma^{2}})\cos^{2}\theta+4}}
\right)
\\
\notag
&=-\tfrac{(\sfrac{\sqrt{\costc}}{\sigma})\sin\theta }{\sqrt{(\sfrac{\costc}{\sigma^{2}})\cos^{2}\theta+4}}\GX.
\end{align}
Then using \eqref{eq:ShannonExp}, we get
\begin{align}
\notag
\der{}{\theta}\SGEX
&=\tfrac{\sqrt{\costc}}{2\sigma}\left[
-\cos\theta\der{}{\theta}\GX+\sin\theta \GX
\right]
-\tfrac{1}{\GX}\der{}{\theta}\GX-\cot\theta
\\
\notag
&=\tfrac{(\sfrac{\sqrt{\costc}}{\sigma})\sin\theta}{\sqrt{(\sfrac{\costc}{\sigma^{2}})\cos^{2}\theta+4}}\GX^{2}
+\tfrac{(\sfrac{\sqrt{\costc}}{\sigma})\sin\theta}{\sqrt{(\sfrac{\costc}{\sigma^{2}})\cos^{2}\theta+4}}
-\cot\theta
\\
\notag
&=\tfrac{\sqrt{\costc}}{\sigma}\sin\theta \GX -\cot\theta.
\end{align}
Thus \eqref{eq:smallangle} and the Taylor's expansion imply
\begin{align}
\notag
\SGEX(\theta_{\blx})
&=\SGEX(\xi)+\left(\tfrac{\sqrt{\costc}}{\sigma}\sin\xi \GX(\xi) -\cot\xi\right)\delta_{\blx}+\bigo{\delta_{\blx}^{2}}.
\end{align}
Invoking  \eqref{eq:delta}, we get the following asymptotic expression 
\begin{align}
\label{eq:ShannonBound-EA-1}
\tfrac{e^{-\blx \SGEX(\theta_{\blx})}}{\sqrt{\blx}}
&\sim\tfrac{e^{-\blx \SGEX(\xi)}}{\sqrt{\blx}}
\left(\sqrt{\tfrac{\pi}{2} \blx}\sin 2\xi \right)^{1-\frac{\sqrt{\costc}}{\sigma}\frac{\sin\xi \GX(\xi)}{\cot\xi}}.
%
\end{align}
On the other hand, invoking \eqref{eq:ShannonGF} in \eqref{eq:ShannonExp},
we get
\begin{align}
\notag
\SGEX(\xi)
&=\tfrac{\costc}{4\sigma^{2}}
\left[2-\cos^{2}\xi+\cos^{2}\xi\sqrt{1+\tfrac{4\sigma^2}{\costc}\tfrac{1}{\cos^{2}\xi}}\right]
-\ln\sin\xi\tfrac{\sqrt{\costc}}{2\sigma}\left(\cos\xi
+\sqrt{\cos^{2}\xi+\tfrac{4\sigma^{2}}{\costc}} \right)
\\
\notag
&=\tfrac{\costc}{4\sigma^{2}}
\left[1+\sin^{2}\xi+(1-\sin^{2}\xi)\sqrt{1+\tfrac{4\sigma^2}{\costc}\tfrac{1}{1-\sin^{2}\xi}}\right]
+\ln\tfrac{1}{\sin\xi} \tfrac{\sqrt{\costc}}{2\sigma}
\left(\sqrt{\cos^{2}\xi+\tfrac{4\sigma^{2}}{\costc}}-\cos\xi\right)
\\
\notag
&=\tfrac{\costc}{4\sigma^{2}}
\left[1+\sin^{2}\xi+(1-\sin^{2}\xi)\sqrt{1+\tfrac{4\sigma^2}{\costc}\tfrac{1}{1-\sin^{2}\xi}}\right]
+\tfrac{1}{2}\ln\tfrac{1}{\sin^{2}\xi} 
\left(
1+\tfrac{\costc\cos^{2}\xi}{2\sigma^{2}}
-\tfrac{\costc\cos^{2}\xi}{2\sigma^{2}}
\sqrt{1+\tfrac{4\sigma^{2}}{\costc \cos^{2}\xi}}\right).
\end{align}
Then using \eqref{eq:phi}, we get the closed-form expression for 
the sphere packing exponent given in
\cite[{(\ref*{D-eq:eg:SGauss-spe})}]{nakiboglu18D}:
\begin{align}
\notag
\SGEX(\xi)
&=\tfrac{\costc}{4\sigma^{2}}
\left[1+\tfrac{1}{e^{2\rate}}+(1-\tfrac{1}{e^{2\rate}})\sqrt{1+\tfrac{4\sigma^2}{\costc}\tfrac{e^{2\rate}}{e^{2\rate}-1}}\right]
+\tfrac{1}{2}\ln\left[e^{2\rate}-\tfrac{\costc (e^{2\rate}-1)}{2\sigma^{2}}
\left(\sqrt{1+\tfrac{4\sigma^2e^{2\rate}}{\costc (e^{2\rate}-1)}}-1\right)\right]
\\
\label{eq:ShannonBound-EA-2}
&=\spe{\rate,\Wm,\costc}.
\end{align}
On the other hand, \eqref{eq:ShannonGF} implies
\begin{align}
\notag
\tfrac{\sqrt{\costc}}{\sigma}\tfrac{\sin\xi \GX(\xi)}{\cot\xi} 
&=\tfrac{\costc\sin^{2}\xi}{2\sigma^{2}}\left(1+\sqrt{1+\tfrac{4\sigma^{2}}{\costc}\tfrac{1}{1-\sin^{2}\xi}}\right)
\\
\notag
&=\tfrac{2\sin^{2}\xi}{1-\sin^{2}\xi}\left(\sqrt{1+\tfrac{4\sigma^{2}}{\costc}\tfrac{1}{1-\sin^{2}\xi}}-1\right)^{-1}.
\end{align}
Invoking first \eqref{eq:phi}, then 
\eqref{eq:gaussian:parametric-rno}, and finally
\eqref{eq:gaussian:SPE-derivative},
we get
\begin{align}
\notag
\tfrac{\sqrt{\costc}}{\sigma}\tfrac{\sin\xi \GX(\xi)}{\cot\xi} 
&=\tfrac{2}{e^{2\rate}-1}
\left(\sqrt{1+\tfrac{4\sigma^{2}}{\costc}\tfrac{e^{2\rate}}{e^{2\rate}-1}}-1\right)^{-1}
\\
\notag
&=\tfrac{1}{\rns}
\\
\notag
&=1-\left.\pder{}{\mS}\spe{\mS,\Wm,\costc}\right\vert_{\mS=\rate}.
\end{align}
Then
\eqref{eq:ShannonBound-EA-1} and 
\eqref{eq:ShannonBound-EA-2} imply
\begin{align}
\notag
\tfrac{e^{-\blx \SGEX(\theta_{\blx})}}{\sqrt{\blx}}
&\sim   
\left(\sqrt{\tfrac{\pi}{2}}\sin 2\xi \right)^{\dspe{\rate}}
\blx^{\frac{\dspe{\rate}-1}{2}}e^{-\blx \spe{\rate,\Wm,\costc}}
\end{align}
where  \(\dspe{\rate}=\left.\pder{}{\mS}\spe{\mS,\Wm,\costc}\right\vert_{\mS=\rate}\).
Then \cite[(3)]{shannon59}, i.e.,\eqref{eq:ShannonsBounds},
imply both \eqref{eq:nonsingular} and \eqref{eq:rSPB}, because
\eqref{eq:ShannonsBounds} 
is nothing but \eqref{eq:nonsingular} and \eqref{eq:rSPB}
for certain multiplicative constants.

\section{Proof of \eqref{eq:bound-on-the-third-moment}}\label{sec:ProofOfeq:bound-on-the-third-moment} 
Note that for any three random variables 
\(\sta_{1}\), \(\sta_{2}\), and \(\sta_{3}\), 
\begin{align}
\notag
\EX{\abs{\sum\nolimits_{\ind=1}^{3}\sta_{\ind}}^{3}}
&\leq\sum\nolimits_{\ind=1}^{3}\EX{\abs{\sta_{\ind}}^3}
+6\EX{\abs{\sta_{1}\sta_{2}\sta_{3}}}
+3\sum\nolimits_{\ind,\jnd:\ind\neq \jnd}
\EX{\abs{\sta_{\ind}\sta_{\jnd}^{2}}}.
\end{align}
On the other hand, as a result of H\"{o}lder's inequality we have
\begin{align}
\notag
\EX{\abs{\sta_{1}\sta_{2}\sta_{3}}}
&\leq
\EX{\abs{\sta_{1}}^{3}}^{\frac{1}{3}}
\EX{\abs{\sta_{2}}^{3}}^{\frac{1}{3}}
\EX{\abs{\sta_{3}}^{3}}^{\frac{1}{3}},
\\
\notag
\EX{\abs{\sta_{\ind}\sta_{\jnd}^{2}}}
&\leq \EX{\abs{\sta_{\ind}}^{3}}^{\frac{1}{3}}
\EX{\abs{\sta_{\jnd}}^{3}}^{\frac{2}{3}}.
\end{align}
Furthermore, since the geometric mean is upper bounded by 
the arithmetic mean, we also have
\begin{align}
\notag
\EX{\abs{\sta_{1}}^{3}}^{\frac{1}{3}}
\EX{\abs{\sta_{2}}^{3}}^{\frac{1}{3}}
\EX{\abs{\sta_{3}}^{3}}^{\frac{1}{3}}
&\leq \tfrac{1}{3}\sum\nolimits_{\ind=1}^{3}\EX{\abs{\sta_{\ind}}^3},
\\
\notag
\EX{\abs{\sta_{\ind}}^{3}}^{\frac{1}{3}}
\EX{\abs{\sta_{\jnd}}^{3}}^{\frac{2}{3}}
&\leq
\tfrac{1}{3}
\EX{\abs{\sta_{\ind}}^{3}}
+\tfrac{2}{3}
\EX{\abs{\sta_{\jnd}}^{3}}.
\end{align}
Thus
\begin{align}
\notag
\EX{\abs{\sum\nolimits_{\ind=1}^{3}\sta_{\ind}}^{3}}
&\leq\sum\nolimits_{\ind=1}^{3}\EX{\abs{\sta_{\ind}}^3}
+2\sum\nolimits_{\ind=1}^{3}\EX{\abs{\sta_{\ind}}^3}
+\sum\nolimits_{\ind,\jnd:\ind\neq \jnd}
\left(\EX{\abs{\sta_{\ind}}^{3}}
+2 \EX{\abs{\sta_{\jnd}}^{3}} \right)
\\
\notag
&\leq 9 \sum\nolimits_{\ind=1}^{3}\EX{\abs{\sta_{\ind}}^3}.
\end{align}
\section*{Acknowledgment}
The author would like to thank Hao-Chung Cheng both 
for numerous inspiring discussions on the sphere packing bound and its refinements, 
which helped the author to simplify and improve the statement of Lemma \ref{lem:HTBE}, 
and for his feedback on the manuscript.
The author would also like to thank the reviewer for pointing out \cite{holevo02}
and for his feedback on the manuscript.

\bibliographystyle{unsrt} 
\newcommand{\noopsort}[1]{} \newcommand{\printfirst}[2]{#1}
  \newcommand{\singleletter}[1]{#1} \newcommand{\switchargs}[2]{#2#1}

\end{document}